\documentclass[journal,onecolumn,12pt]{IEEEtran}
\usepackage{amsmath}
\usepackage{ntheorem}

\newcommand*{\QEDA}{\hfill\ensuremath{\blacksquare}}  
\newtheorem*{proof}{{\textit{Proof:}}}
\newtheorem{lemma}{{\textit{\textbf{Lemma}}}}
\newtheorem{theorem}{{\textit{\textbf{Theorem}}}}
\usepackage{graphicx}
\usepackage{epstopdf}
\usepackage{verbatim}
\usepackage{upgreek}
\usepackage{amssymb,amsmath}
\usepackage{color}
\usepackage{epstopdf}
\usepackage{bm}
\usepackage{cite}
\usepackage{array,color}
\usepackage{algorithm}
\usepackage{algpseudocode}
\usepackage{amsmath}
\usepackage{graphics}
\usepackage{epsfig}
\usepackage{ulem}
\usepackage[switch]{lineno}
\usepackage{threeparttable}
\usepackage{multicol,multirow}
\usepackage{soul}
\usepackage{diagbox}
\usepackage{subfigure}
\usepackage{xcolor}
\ifCLASSINFOpdf
\else
\fi


\begin{document}
\title{Mutual Information Analysis for Factor Graph-based MIMO Iterative Detections through Error Functions}

\author{
\IEEEauthorblockA{Huan Li, Jingxuan Huang, and Zesong Fei
\thanks{This work has been submitted to the IEEE for possible publication. Copyright may be transferred without notice, after which this version may no longer be accessible.}
\thanks{H. Li, J. Huang and Z. Fei are with the School of Information and Electronics, Beijing Institute of Technology, Beijing 100081, China (e-mail: huan\_lilh@outlook.com and 3120200780@bit.edu.cn, jxhbit@gmail.com, feizesong@bit.edu.cn.).}\thanks{}
}}



\maketitle

\begin{abstract}
	The factor graph (FG) based iterative detection is considered an effective and practical method for multiple-input and multiple-out (MIMO), particularly massive MIMO (m-MIMO) systems. However, the convergence analysis for the FG-based iterative MIMO detection is insufficient, which is of great significance to the performance evaluation and algorithm design of detection methods. This paper investigates the mutual information update flow for the FG-based iterative MIMO detection and proposes a precise mutual information computation mechanism with the aid of Gaussian approximation and error functions, i.e., the error functions-aided analysis (EF-AA) mechanism. Numerical results indicate that the theoretical result calculated by the EF-AA mechanism is completely consistent with the bit error rate performance of the FG-based iterative MIMO detection. Furthermore, the proposed EF-AA mechanism can reveal the exact convergent iteration number and convergent signal-to-ratio value of the FG-based iterative MIMO detection, representing the performance bound of the MIMO detection.
\end{abstract}
\IEEEpeerreviewmaketitle


\section{Introduction}
Ultra-high speed and ultra-reliable wireless transmissions are requested in the sixth-generation (6G) technology to provide ubiquitous high-performance connections \cite{1:6G}. Therefore, multiple-input and multiple-output (MIMO) technology \cite{2:MIMO} is significant to the 6G implementation, which can guarantee increased data throughput with accurate detection methods. Specifically, massive MIMO (m-MIMO) technology \cite{3:m-MIMO} can provide an extremely high data transmission rate due to more transceiver antennas and spatial diversity. 

However, most detection algorithms are not feasible for m-MIMO technology because of the inevitable high complexity \cite{4:ML,5:LR,6:LRa}. To solve this problem, the authors in \cite{7:FGMP} proposed to utilize the factor graph (FG) model to transfer probability information between observation nodes (ONs) and variable nodes (VNs) to estimate accurate symbol probabilities. S. Wu et al. in \cite{8:FGEP} showed that FG-based iterative MIMO detection could achieve near-optimal performance compared to the optimal maximum likelihood detection \cite{4:ML}, with a complexity that is acceptable and quadratic to the number of transceiver antennas \cite{9:FGEPCMP}. Consequently, the FG-based iterative MIMO detection can be considered a practical method for m-MIMO technology.

Specifically, the performance of MIMO detections can be influenced by many factors, for example, the property of MIMO channels. Many prior arts have analyzed the MIMO channel through mutual information \cite{10:UNCSI,11:IIDMIMO,12:CORREI,13:mMIMOAna,14:MATADC}. In \cite{10:UNCSI}, O. Oyman et al. derived a tight lower-bounded analytical expression for a Gaussian MIMO frequency-selective spatially correlated fading channel with unknown channel state information (CSI), which approximated the variance of mutual information to a closed-form function. For the space-time independent and identically distributed (i.i.d.) MIMO channel, the authors in \cite{11:IIDMIMO} proposed analytical expressions to present the distribution characters of mutual information between transmitted and received signal vectors. L. Musavian et al. proved in \cite{12:CORREI} that a tight gap between the upper and lower bounds of mutual information exists when uncorrelated transmit antennas with uniform power distribution and correlated receive antennas are considered. For the m-MIMO systems, P. Yang et al. have developed a message-passing-based algorithm to compute the mutual information where arranged finite-alphabet inputs \cite{13:mMIMOAna}. Recently, random matrix theory has been utilized to attain the ergodic mutual information between the transmit signals and outputs of the Rayleigh channel \cite{14:MATADC}, which is quantized by a mixed analog-to-digital converters architecture. 

Based on the mutual information investigations of different MIMO channels, much literature has studied the performance of varieties of MIMO detections from the perspective of mutual information. For the serial detection schemes in vertical Bell labs layered space-time (V-BLAST) MIMO architecture \cite{15:BLAST}, S. Stiglmayr et al. calculated the mutual information of a single antenna stream without considering the correlation between the Gaussian noise and hard decision error \cite{16:VBLAST}. In \cite{17:STAIR}, the authors proposed a numerical calculation of the mutual information to show the convergence of an iterative method, which utilized the stair matrix to achieve similar performance to linear minimum mean-square error (MMSE) detection. Particularly for the belief propagation-based iterative detections, \cite{18:EXITMIMO} employed the extrinsic information transfer (EXIT) chart to present the validity and convergence of the iterative process. Then, we initially derived the closed-form mutual information update flow for the FG-based iterative MIMO detection according to the EXIT analysis in \cite{19:EXITMIMO2}. Nevertheless, the proposed mutual information calculations in \cite{19:EXITMIMO2} ignored the specific probability when computing the mutual information between VN and the transmitted symbol. The existing approximation of mutual information causes inaccuracy in evaluating the convergence and bit error rate (BER) performance.

Therefore, we derive the closed-form expressions of assessing the convergence and BER performance of the FG-based iterative MIMO detection, which is more feasible for the MIMO, especially the m-MIMO system. In this paper, we proposed a more precise calculation method for the mutual information update flow of the FG-based iterative MIMO detection under both binary phase shift keying (BPSK) and quadrature phase shift keying (QPSK) modulations, where the error function $\mathit{erf}(\cdot)$ and the complementary error function $\mathit{erfc}(\cdot)$ are utilized. Our proposed error functions-aided mechanism can provide exact mutual information curves at different SNRs in both MIMO and m-MIMO systems, which are of great significance to theoretical bound evaluation and the detection algorithm designs. 

Noted that this work is distinct from the aforementioned studies \cite{10:UNCSI,11:IIDMIMO,12:CORREI,13:mMIMOAna,14:MATADC}, which focuses on precisely investigating the performance of FG-based iterative MIMO detections, instead of MIMO channel properties. The rest of this paper is organized as follows. In Section II, the system model and fundamental knowledge of FG-based iterative MIMO detections are introduced. In Section III, we derived the proposed error functions-aided mutual information calculation mechanism for FG-based iterative MIMO detections. The numerical results of mutual information analysis and the BER performance are presented in Section IV. Finally, we conclude this paper in Section V.

\section{Preliminaries}
\subsection{Channel Model}
In this paper, we consider a MIMO system equipped with numerous antennas at both the transmit and receive sides, represented by $N_T$ and $N_R$, respectively. Note that the number of antennas can range from 2 to hundreds. Specifically, the received signal vector $\boldsymbol{y}\in\mathbb{C}^{{N_R}\times 1}$ can be given as
\begin{equation}\label{equ:yR}
\boldsymbol{y} = \boldsymbol{Hx} + \boldsymbol{n},
\end{equation}
where $\boldsymbol{x}=\left[{x_1,x_2,\cdots,x_{N_T}}\right] \in\mathbb{C}^{{N_T}\times 1}$ denotes the power-normalized transmitted signal vector, $\boldsymbol{n} = \left[{n_1,n_2,\cdots,n_{N_R}}\right] \in\mathbb{C}^{{N_R}\times 1}$ denotes the additive complex-valued Gaussian white noise, which can be represented as $n_i\sim {\cal{CN}}(0,\sigma_n^2)$. The channel matrix $\boldsymbol{H}\in\mathbb{C}^{{N_R}\times {N_T}}$ reflects the Rayleigh fading effects and can be expressed as
\begin{equation}\label{equ:HCOM}
\begin{aligned}{
	{\boldsymbol{H}} = \left[ {\begin{array}{*{20}{c}}
		{{h_{1,1}}}& \cdots &{{h_{1,N_T}}}\\
		\vdots & \ddots & \vdots \\
		{{h_{N_R,1}}}& \cdots &{{h_{N_R,N_T}}}
		\end{array}} \right]
},
\end{aligned}
\end{equation}
where entries follow the complex-valued Gaussian distribution with zero mean and unit variance. 

Specifically, the V-BLAST MIMO structure \cite{15:BLAST} is adopted at the transmit and receive sides, shown as Fig. \ref{fig:VBLAST}. To evaluate the performance of FG-based iterative MIMO detections more precisely, the averaged received SNR $\rho_r$ is given as
\begin{equation}
{\rho_r} =\mathbb{E}\left\{\frac{{\sum\nolimits_{i = 1,j = 1}^{i = N_R,j = N_T} {|{h_{i,j}}{|^2}} }}{N_T{\sigma _n^2}}\right\}.
\end{equation}
\begin{figure}
	\centering
	\includegraphics[width=0.5\textwidth]{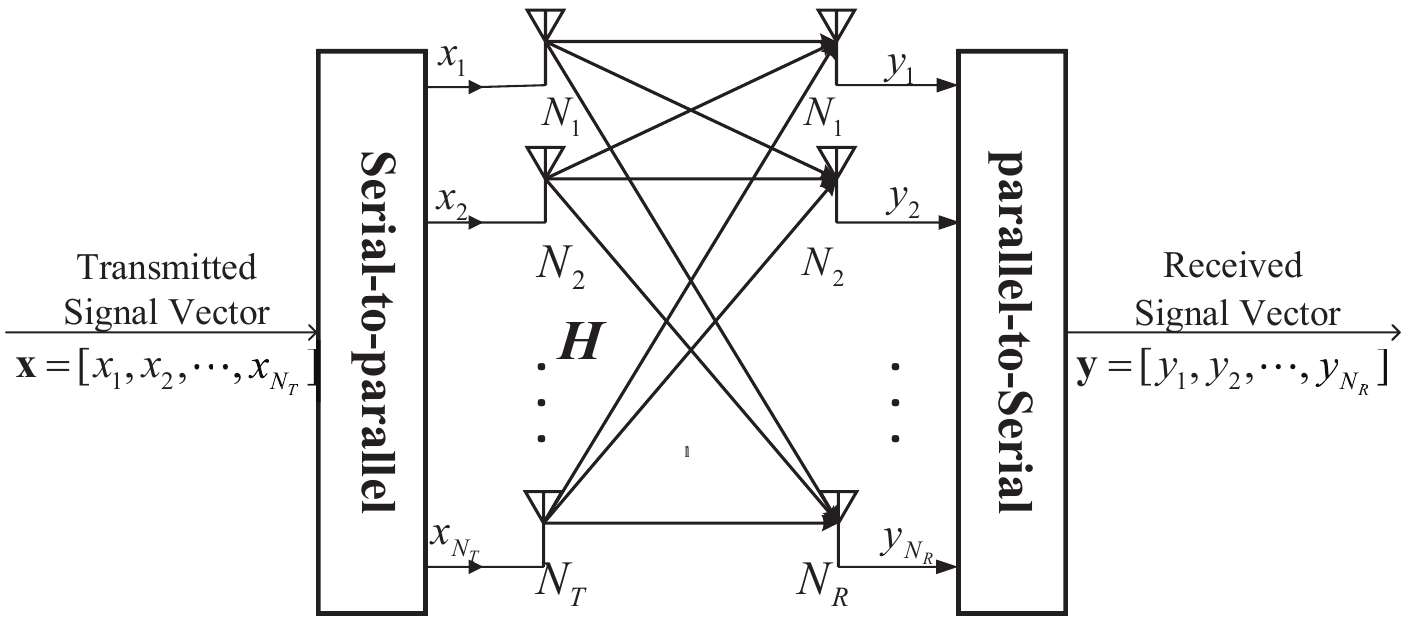}
	\caption{V-BLAST architecture.}\label{fig:VBLAST}
\end{figure}
\subsection{An Overview of the FG-based Iterative MIMO Detection}
The FG-based MIMO detection is a kind of message-passing algorithm, which transmits \textit{a posteriori} probability information between ONs $o_i$ $(i=1,2,\cdots,N_R)$ and VNs $v_l$ $(l=1,2,\cdots,N_T)$. The specific information transfer flow in FG is given in Fig. \ref{fig:FGEP}. 
\begin{figure}
	\centering
	\includegraphics[width=0.5\textwidth]{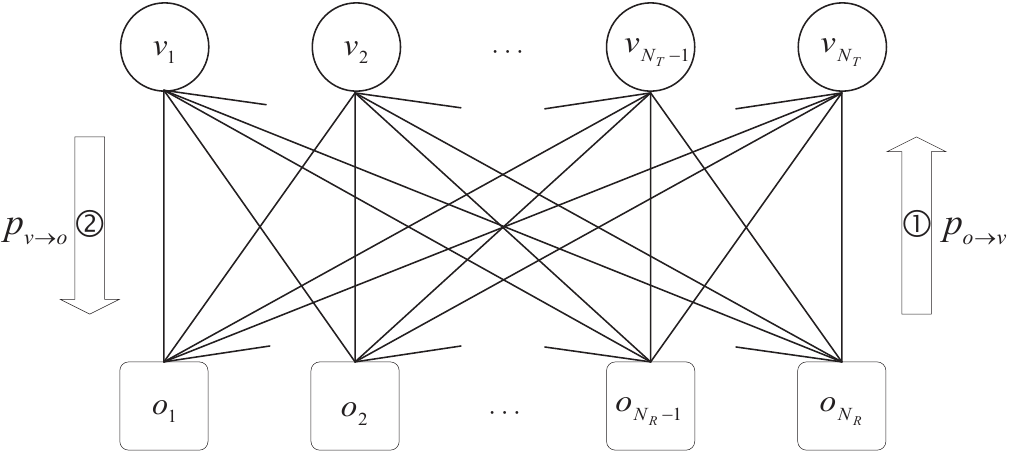}
	\caption{Information transfer flow in FG-based iterative MIMO detections.}\label{fig:FGEP}
\end{figure}

Shown as in Fig. \ref{fig:FGEP}, the FG-based iterative MIMO detection firstly transmit the probability information $p_{o\to v}$ from ON to VN, which is given as
\begin{equation}
p_{o\to v} = \prod\nolimits_{v' \in V(o)\backslash v} {{p_{v' \to o}}},
\end{equation}
where $V(o)\backslash v$ denotes the collection of VN connected to ON except VN $v$.

Then the probability information $p_{v \to o}$ transferred from VN to ON can be expressed as
\begin{equation}
p_{v\to o} = \prod\nolimits_{o' \in O(v)\backslash o} {{p_{o' \to v}}},
\end{equation}
where $O(v)\backslash o$ represents the collection of ON linked to VN except ON $o$. 

To conclude, the FG-based iterative MIMO detection follows the bidirectional-transmission mechanism until the \textit{a posteriori} probability information stays unchanged, which is defined as convergence \cite{20:Conver}.

\section{Error Functions-Aided Analysis Mechanism for FG-based Iterative MIMO Detections}
In our previous work \cite{19:EXITMIMO2}, we introduced an innovative EXIT analysis method with the ability to evaluate mutual information through low-complexity calculations. Although the earlier EXIT analysis was able to present the convergence of iterative MIMO detections, it was not accurate enough. In this paper, we propose a new method that can generate mutual information curves of iterative MIMO detections more precisely.

\begin{figure}
	\centering
	\includegraphics[width=0.5\textwidth]{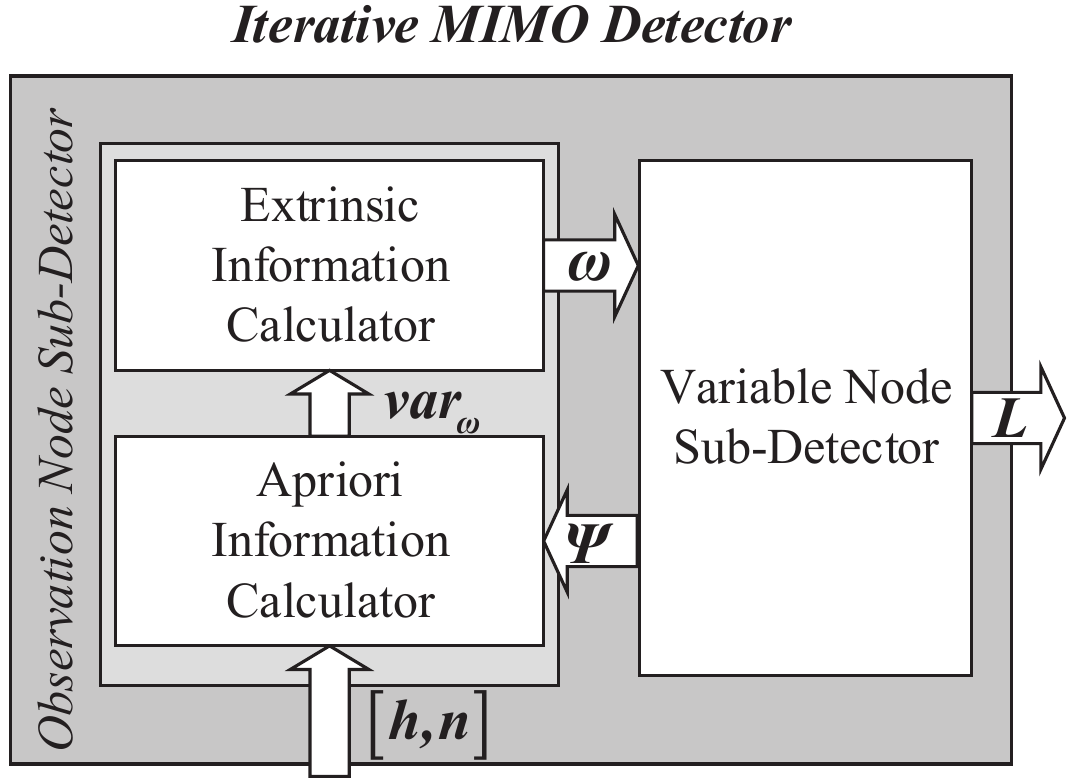}
	\caption{Extrinsic information transfer flow for iterative MIMO detections.}\label{fig:EXITFig}
\end{figure}

The information transfer flow of FG-based iterative MIMO detections can be demonstrated in Fig. \ref{fig:EXITFig}, where the iterative MIMO detector is abstracted to be composed of the ON sub-detector and VN sub-detector. Specifically, the ON sub-detector concludes two components, i.e., the extrinsic information calculator (EIC) and the apriori information calculator (AIC). 

As shown in Fig. \ref{fig:EXITFig}, the AIC firstly updates the information transferred to it, utilizing channel information $h_{i,l}$, $n_i$ ($i=1,2,\cdots,N_R$, $l=1,2,\cdots,N_T$) transformed by VN sub-detector. Hence, the output information $\mathit{var_{\omega_i^l}}$ of AIC can be expressed as 
\begin{equation}\label{equ:defvaromega}
\begin{split}
\mathit{var_{\omega_i^l}} &= \mathit{F_{AIC}}\left({c_{i,l},\psi_i^l}\right),\\
for \: i&=1,2,\cdots,N_R, \:and\: j = 1,2,\cdots,N_T, 
\end{split}
\end{equation}
where $\mathit{F_{AIC}}\left({\cdot}\right)$ denotes the transfer function of AIC.

Another component of the ON sub-detector is EIC, which calculates the mutual information $I_{\omega_i^l}$ between transmit symbol $x_l$ and LLR $\omega_i^l$, and then transfer the mutual information $I_{\omega_i^l}$ to VN sub-detector. The mutual information $I_{\omega_i^l}$ can be represented by (\ref{equ:defIomega}),
\begin{figure*}[h]
	\begin{equation}\label{equ:defIomega}
	\begin{split}
	I_{\omega_i^l} = I({x_l};\omega _i^l)\buildrel \Delta \over =& \frac{1}{2}\sum\limits_{\theta =  \theta_1,\theta_2 }{\int_{ - \infty }^{ + \infty } {{p_\omega }(\omega _i^l|{x_l} = \theta)} } {\log _2}\left( {\frac{{2{p_\omega }(\omega _i^l|{x_l} = \theta)}}{{{p_\omega }(\omega _i^l|{x_l} = \theta_1 ) + {p_\omega }(\omega _i^l|{x_l} =  \theta_2 )}}} \right)d\omega _i^l\\
	\buildrel \Delta \over =& \mathit{F_{\omega}}\left({\mathit{var_{\omega_i^l}}}\right), 
	\end{split}
	\end{equation}
\end{figure*}
where ${p_\omega }(\omega _i^l|{x_l})$ denotes the conditional probability density function (CPDF) of the LLR $\omega_i^l$ at ON $o_i$. In addition, $\mathit{F_{\omega}}\left({\cdot}\right)$ denotes the transfer function of VN-detector, which also indicates that the CPDF ${p_\omega }(\omega _i^l|{x_l})$ is relevant to the output information of AIC $\mathit{var_{\omega_i^l}}$.

For the VN sub-detector, mutual information $I_{\psi_i^l}$ between transmit symbol $x_l$ and LLR $\psi_i^l$ is related to the output information of EIC $\omega_i^l$, and can be expressed by \ref{equ:defIpsi}, which is then transferred to ON sub-detector.
\begin{figure*}[h]
	\begin{equation}\label{equ:defIpsi}
	\begin{split}
	I_{\psi_i^l} = I({x_l};\psi _i^l)\buildrel \Delta \over =& \frac{1}{2}\sum\limits_{\theta =  \theta_1,\theta_2 }{\int_{ - \infty }^{ + \infty } {{p_\psi }(\psi _i^l|{x_l} = \theta)} } {\log _2}\left( {\frac{{2{p_\psi }(\psi _i^l|{x_l} = \theta)}}{{{p_\psi }(\psi _i^l|{x_l} = \theta_1 ) + {p_\psi }(\psi _i^l|{x_l} =  \theta_2 )}}} \right)d\psi _i^l\\
	\buildrel \Delta \over =& \mathit{F_{\psi}}\left({\mathit{var_{\psi_i^l}}}\right), 
	\end{split}
	\end{equation}
	\hrulefill
\end{figure*}
In (\ref{equ:defIpsi}), ${p_\psi }(\psi _i^l|{x_l})$ represents the CPDF of the LLR $\psi_i^l$ at VN $v_l$. Similarly, $\mathit{F_{\psi}}\left( \cdot\right) $ denotes the transfer function of VN sub-detector, and the CPDF ${p_\psi }(\psi _i^l|{x_l})$ is related to the output information of EIC ${\omega_i^l}$. Furthermore, the output information $L_l=\sum_{i'=1}^{N_R}{\omega_{i'}^l}$ of the iterative MIMO detection is also relevant to the output information of EIC ${\omega_i^l}$, then the mutual information of which can be given as
\begin{equation}
I_{L_l} = \mathit{F_{L}}\left({{\mathit{var_{\omega_i^l}}}}\right),
\end{equation}
where $\mathit{F_{L}}\left( \cdot\right)$ represents the transfer function of iterative MIMO detector.

\subsection{The Error Functions-Aided Analysis under BPSK Modulation}
This paper first derives the analysis mechanism for FG-based iterative MIMO detections under BPSK modulation. The previous work \cite{19:EXITMIMO2} has derived the approximate mutual information of the ON sub-detector and VN sub-detector. Although the analysis in \cite{19:EXITMIMO2} is imperfect, it still provides some significant derivation for our current work, which are concluded as {\textbf{\textit{Lemma} \ref{lem:1}}} to {\textbf{\textit{Lemma} \ref{lem:3}}} as follows.

\begin{lemma}\label{lem:1}
	\textup{The mutual information $I_{\omega_i^l}$ between transmited symbol $x_l$ and LLR $\omega_i^l$ at ON $o_i$ can be expressed as}
	\begin{equation}\label{equ:Iomega}
	{I_{\omega_i^l}} \buildrel \Delta \over = \mathit{F_{\omega}}\left({\mathit{var_{\omega_i^l}}}\right)\buildrel \Delta \over = J(\sqrt{\mathit{var_{\omega_i^l}}}),
	\end{equation}
	\textup{where $J\left(\cdot\right)$ is a curve fitting function for the integral in (\ref{equ:defIomega}) proposed in \cite{21:Jfunc}, which is given in Appendix A.}
	
	\textup{Specifically, $\mathit{var_{\omega_i^l}}$ was defined as the equivalent channel variable (ECV) \cite{19:EXITMIMO2} transferred by the AIC, which is computed based on the channel information integrated with extrinsic information from VN sub-detector as follows}
	\begin{equation}\label{equ:ECVori}
	\mathit{var_{\omega_i^l}} = \frac{{4h_{i,l}^2}}{{\sigma _{{g_{il}}}^2}},
	\end{equation} 
	\textup{where ${\sigma _{{g_{il}}}^2}$ denotes the variance of interference signals (i.e., signals from other antennas sum up with channel noise)}.
	
	\textup{The details can be seen in \cite{19:EXITMIMO2}.}
\end{lemma}

\begin{lemma}\label{lem:2}
	\textup{The CPDF of extrinsic LLR $\psi_i^l$ at VN $v_l$ follows the Gaussian distribution of mean $\mu_{\psi_i^l}$ and variance $\mathit{var_{\psi_i^l}}$, which can be represented as}
	\begin{equation}\label{equ:PDFpsi}
	\begin{split}
	{p_\psi }(\psi _i^l|{x_l}) =& {p_\psi }(\psi _i^l|{h_{i,l}},{x_l})\\
	=& \frac{1}{{\sqrt {2\pi \mathit{var_{\psi_i^l}}} }}\exp \left( { - \frac{{{{\left( {\psi _i^l - \mu_{\psi_i^l} } \right)}^2}}}{{2\mathit{var_{\psi_i^l}}}}} \right),
	\end{split}
	\end{equation}
	\textup{where $\mathit{var_{\psi_i^l}}$ denotes the variance of the extrinsic LLR $\psi _i^{l}$, and can be calculated by}
	\begin{equation}\label{equ:varpsi}
	\begin{split}
	\mathit{var_{\psi_i^l}}=&\sum\nolimits_{i'=1,i'\ne i}^{N_R}{\mathit{var_{\omega_{i'}^l}}}\\
	=&\sum\nolimits_{i'=1,i'\ne i}^{N_R}{J^{-1}\left({I_{\omega_i^{l'}}}\right)},
	\end{split}
	\end{equation}
	\textup{where $J^{-1}\left({\cdot}\right)$ denotes the inverse function of the curve fitting function $J\left(\cdot\right)$, which is given in Appendix A. Furthermore, the mean value is given as}
	\begin{equation}
	\begin{split}
	\mu_{\psi_i^l} =& \pm\frac{\mathit{var_{\psi_i^l}}}{2}\\
	=& \sum\nolimits_{i'=1,i'\ne i}^{N_R}{\pm\frac{\mathit{var_{\omega_{i'}^l}}}{2}}.
	\end{split}
	\end{equation}
	
	\textup{The details can be seen in \cite{19:EXITMIMO2}.}
\end{lemma}

\begin{lemma}\label{lem:3}
	\textup{The mutual information $I_{\psi_i^l}$ between transmited symbol $x_l$ and LLR $\psi_i^l$ at VN $v_l$ is given as}
	\begin{equation}\label{equ:Ipsi}
	{I_{\psi _i^{l}}} \buildrel \Delta \over = \mathit{F_{\psi}}\left({\mathit{var_{\psi_i^l}}}\right)= J\left( {\sqrt{\mathit{var_{\psi_i^l}}}} \right).
	\end{equation}
	
	\textup{The details can be seen in \cite{19:EXITMIMO2}.}
\end{lemma}

In what follows, we derive the mutual information ${I_{\omega_i^l}}$, ${I_{\psi _i^{l}}}$ and $I_{L_l}$ according to Fig. \ref{fig:EXITFig} and the aforementioned {\textbf{\textit{Lemmas}}}, respectively. It is noted that the mutual information ${I_{\omega_i^l}}$ calculated by EIC at ON sub-detector is given in {\textbf{\textit{Lemma} \ref{lem:1}}} as (\ref{equ:Iomega}), with parameter $\mathit{var_{\omega_i^l}}$ as the input. Furthermore, the mutual information ${I_{\psi _i^{l}}}$ computed at VN sub-detector is defined in {{\textbf{\textit{Lemma}}} \textbf{\ref{lem:3}}} as (\ref{equ:Ipsi}), and the input parameter $\mathit{var_{\psi_{i}^l}}$ can be obtained through the summation of $\mathit{var_{\omega_i^l}}$. Therefore, in this section, we mainly focus on the computation of the variance $\mathit{var_{\omega_i^l}}$ at the AIC, i.e., the function $\mathit{F_{AIC}}\left({\cdot}\right)$ in (\ref{equ:defvaromega}).

At AIC, the received symbol $y_i$ can be separated into two parts, namely the symbol from the corresponding transmit antenna and interference, which is given as 
\begin{equation}\label{equ:GA}
\begin{split}
{y_i} &= \sum\nolimits_{l = 1}^{N_T} {{h_{i,l}}{x_l}}  + {n_i}\\
&= {h_{i,l}}{x_l} + \underbrace {\sum\nolimits_{l' = 1,l' \ne l}^{N_T} {{h_{i,l'}}{x_{l'}}}  + {n_i}}_{{g_{il}}},
\end{split}
\end{equation}
where the signals from other antennas plus channel noise are approximated to Gaussian random variables as ${g_{il}}\sim{\cal N}({\mu _{{g_{il}}}},\sigma _{{g_{il}}}^2)$. Specifically, the mean ${\mu _{{g_{il}}}}$ and variance $\sigma _{{g_{il}}}^2$ can be computed as
\begin{equation}\label{equ:meanG}
{\mu _{{g_{il}}}} = \sum\nolimits_{l' = 1,l' \ne l}^{N_T} {{h_{i,l'}}{\mathbb{E}(x_{l'})}} ,
\end{equation}
\begin{equation}\label{equ:varG}
\begin{split}
\sigma _{{g_{il}}}^2 =& \sum\nolimits_{l' = 1,l' \ne l}^{N_T}{{\left| h_{i,l'}\right| ^2}{\mathbb{V}(x_{l'})} }+\sigma_n^2,
\end{split}
\end{equation}
in which
\begin{equation}\label{equ:Exl}
\mathbb{E}(x_{l'})={\sum\nolimits_{j = 1}^2 {{P_i^{l'}}({\theta_j})\cdot{\theta_j}}},
\end{equation}
and 
\begin{equation}\label{equ:Varxl}
\mathbb{V}(x_{l'})={ {\sum\nolimits_{j = 1}^2 {{P_i^{l'}}({\theta_j})\cdot{\theta_j}^2}}-\left| {\mathbb{E}(x_{l'})}\right|^2},
\end{equation}
are the mean and variance of $x_{l'}$ respectively, in which $P_i^{l'}(\theta_j)$ represents the probability of $x_{l'}=\theta_j\in\left\lbrace{\pm 1} \right\rbrace $ estimated at VN $v_{l'}$.

From {{\textbf{{\textbf{\textit{Lemma}}} \ref{lem:1}}}}, we get to know that the ECV $\mathit{var_{\omega_i^l}}$ is essential to the calculation of mutual information $I_\omega$. Formula (\ref{equ:ECVori}) also indicates that the critical part of ECV $\mathit{var_{\omega_i^l}}$ is the variance $\sigma_{{g_{il}}}$. Furthermore, according to (\ref{equ:varG}) and (\ref{equ:Varxl}), the calculation of probability $P_i^{l'}(\theta_j)$ at VN $v_l$ is indiepensable for the computation of variance $\sigma_{{g_{il}}}$, as well as the ECV $\mathit{var_{\omega_i^l}}$. 

Specifically, the $P_i^{l}(\theta_j)$ at VN $v_l$ can be initialized as an equal probability in the first iteration during the iterative MIMO detections. Therefore, in the first iteration, we have 
\begin{equation}\label{equ:sigit1}
\begin{split}
\mathbb{V}(x_{l})=& \frac{1}{2}\left( {1 ^2+\left( {-1}\right) ^2}\right)\!-\!\left| \frac{1}{2}\left( {{1}-{1}}\right) \right|^2\\
=& 1.
\end{split} 
\end{equation}

However, the variance $\mathbb{V}(x_{l})$ remains uncertain in the subsequent iterations, for the reason that $P_i^{l'}(\theta_j)$ is changeable with the LLR $\omega_{i}^l$ transferred from the ON sub-detector. Consequently, it is crucial to derive the probability $P_i^{l'}(\theta_j)$ at VN $v_l$ in the subsequent iterations. Before that, we propose the following {{\textbf{{\textbf{\textit{Theorem}}} \ref{the:1}}}} as the basis of our further analysis.
\begin{theorem}\label{the:1}
	\textup{Given the PDF $p_{L}\left( {L_\alpha}\right) $ of LLR $L_\alpha$ for symbol $\alpha$, we can get the probability of symbol $\alpha = 1$ as}
	\begin{equation}\label{equ:P_1}
	P(\alpha = 1) = \int_0^\infty  {{p_L}\left( {{L_\alpha }} \right)d{L_\alpha }}.
	\end{equation}
	
	\textup{Similarly, the probability of symbol $\alpha = 0$ is given as}
	\begin{equation}\label{equ:P_0}
	P(\alpha = 0) = \int_{-\infty}^0  {{p_L}\left( {{L_\alpha }} \right)d{L_\alpha }}.
	\end{equation}
\end{theorem}
\begin{proof}
	\textup{In the process of iterative MIMO detection, hard decision adopts the following criteria}
	\begin{equation}\label{equ:hardes}
	\alpha  = \left\{ {\begin{array}{*{20}{l}}
		{1,}&{if\:{L_\alpha } > 0}\\
		{0,}&{if\:{L_\alpha } < 0}
		\end{array}}, \right.
	\end{equation}
	\textup{which indicates the value of $\alpha$ depends on the sign of its LLR ${L_\alpha }$. While given the probability of the symbol's LLR, instead of its accurate value, we can only obtain the probability of $P({L_\alpha } > 0)$ and $P({L_\alpha } < 0)$ through the integral of $p_{L}\left( {L_\alpha}\right)$, which can be expressed as}
	\begin{equation}\label{equ:PLLA}
	P({L_\alpha } > 0) = \int_0^\infty  {{p_L}\left( {{L_\alpha }} \right)d{L_\alpha }},
	\end{equation}
	\textup{and}
	\begin{equation}\label{equ:PLLE}
	P({L_\alpha } < 0) = \int_{-\infty}^0  {{p_L}\left( {{L_\alpha }} \right)d{L_\alpha }}.
	\end{equation}
	
	\textup{Thus, combined with (\ref{equ:PLLA}) and (\ref{equ:PLLE}), which reflect the sign of ${L_\alpha }$ through probability, (\ref{equ:hardes}) can be revised as}
	\begin{equation}\label{equ:P}
	\left\{ {\begin{array}{*{20}{l}}
		{P(\alpha = 1)\propto P({L_\alpha } > 0)}\\
		{P(\alpha = 0)\propto P({L_\alpha } < 0)}
		\end{array}}. \right.
	\end{equation}
	
	\textup{Therefore, we can represent the probability of $\alpha$ as (\ref{equ:P_1}) and (\ref{equ:P_0}), respectively.}
	\QEDA
\end{proof}

In {{\textbf{{\textbf{\textit{Lemma}}} \ref{the:2}}}}, the CPDF of extrinsic LLR $\psi_i^l$ at VN $v_l$ is given as a Gaussian PDF with known mean and variance. Then we can derive the probability $P_i^{l'}(\theta_j)$ at VN $v_l$ in the subsequent iterations according to {{\textbf{{\textbf{\textit{Theorem}}} \ref{the:1}}}} and the CPDF (\ref{equ:PDFpsi}) in {{\textbf{{\textbf{\textit{Lemma}}} \ref{the:2}}}}.

\begin{theorem}\label{the:2}
	\textup{Given the variance $\mathit{var_{\psi_{i}^l}}$ of extrinsic LLR $\psi_i^{l}$ of VN sub-detector, we can calculate the probability $P_i^{l}(\theta_j)$ estimated at VN $v_{l}$ in the subsequent iterations as follows}
	\begin{equation}\label{equ:PilP}
	\begin{split}
	P_i^l( + 1 ) =& P_i^l(\theta_j =  + 1 )\\
	=& \left\{ {\begin{array}{*{20}{l}}
		{\frac{1}{2}\left[1+\mathit{erf}\left( {\sqrt {{\mathit{var_{\psi_i^l}}}/{8}} } \right)\right] ,if\:{x_l} =  + 1 ,}\\
		{\frac{1}{2}\mathit{erfc}\left( {\sqrt {{\mathit{var_{\psi_i^l}}}/{8}} } \right),if\:{x_l} =  - 1,}
		\end{array}} \right.
	\end{split}
	\end{equation}
	\textup{and}
	\begin{equation}\label{equ:PilW}
	\begin{split}
	P_i^{l}(- 1) =& P_i^l(\theta_j =  - 1 )\\
	=& \left\{ {\begin{array}{*{20}{l}}
		{\frac{1}{2}\left[1+\mathit{erf}\left( {\sqrt {{\mathit{var_{\psi_i^l}}}/{8}} } \right)\right],if\:{x_l} =  - 1 ,}\\
		{\frac{1}{2}\mathit{erfc}\left( {\sqrt {{\mathit{var_{\psi_i^l}}}/{8}} } \right),if\:{x_l} =  + 1,}
		\end{array}} \right.
	\end{split}
	\end{equation}
	\textup{where $\mathit{erf}\left({\cdot}\right) $ and $\mathit{erfc}\left({\cdot}\right)$ are the error function and complementary error function respectively, which were defined in \cite{22:ERF} as}
	\begin{equation}
	\mathit{erf}\left({x}\right) = \frac{2}{{\sqrt \pi  }}\int_0^x  {{exp\left( { - {t^2}}\right) }dt},
	\end{equation}
	\textup{and}
	\begin{equation}
	\mathit{erfc}\left({x}\right) = 1 - \mathit{erf}\left({x}\right)\\
	= \frac{2}{{\sqrt \pi  }}\int_x^\infty  {{exp({ - {t^2}})}dt}.
	\end{equation}
\end{theorem}
\begin{proof}
	\textup{From aforementioned {{\textbf{{\textbf{\textit{Lemma}}} \ref{the:2}}}} we can know that the CPDF estimated at VNs follows the Gaussian distribution, and is given as}
	\begin{equation}
	\begin{split}
	{p_\psi }(\psi _i^l|{h_{i,l}},{x_l}) =& {p_\psi }(\psi _i^l|{{x_l}=\mp 1})\\
	=& \frac{1}{{\sqrt {2\pi \mathit{var_{\psi_{i}^l}}} }}\exp \left( { - \frac{{{{\left( {\psi _i^l \pm  \mathit{var_{\psi_{i}^l}}/2} \right)}^2}}}{{2\mathit{var_{\psi_{i}^l}}}}} \right),
	\end{split}
	\end{equation}
	\textup{which is determined by the transmit symbol $x_l$.}
	\begin{figure}
		\centering
		\includegraphics[width=0.5\textwidth]{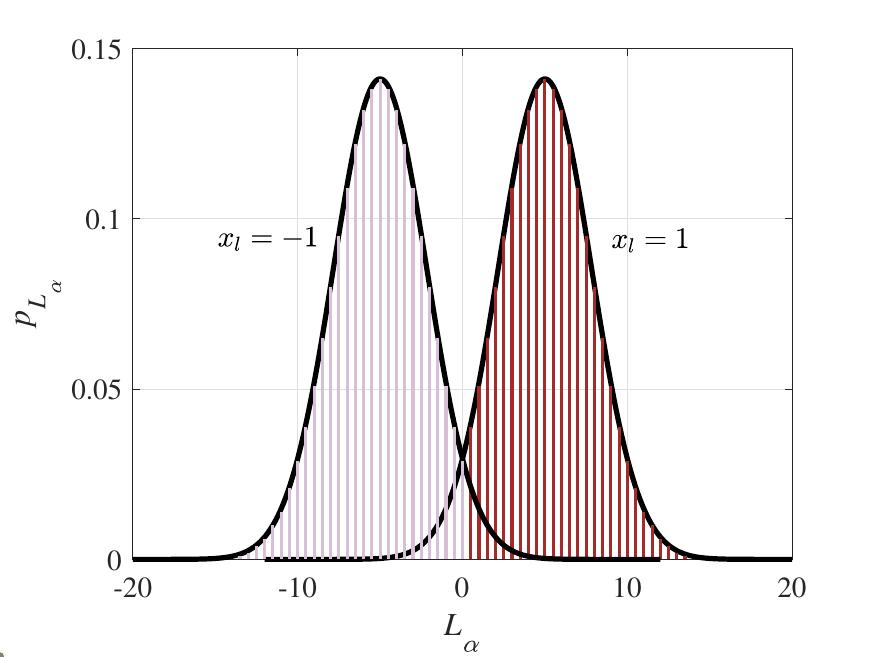}
		\caption{The probability of right-decision $P_i^l({\theta _j} \to R)$ represented by the integral of ${p_\psi }(\psi _i^l|{x_l=\mp 1})$ (see the shadowing parts).}\label{fig:PILR}
	\end{figure}
	\begin{figure}
		\centering
		\includegraphics[width=0.5\textwidth]{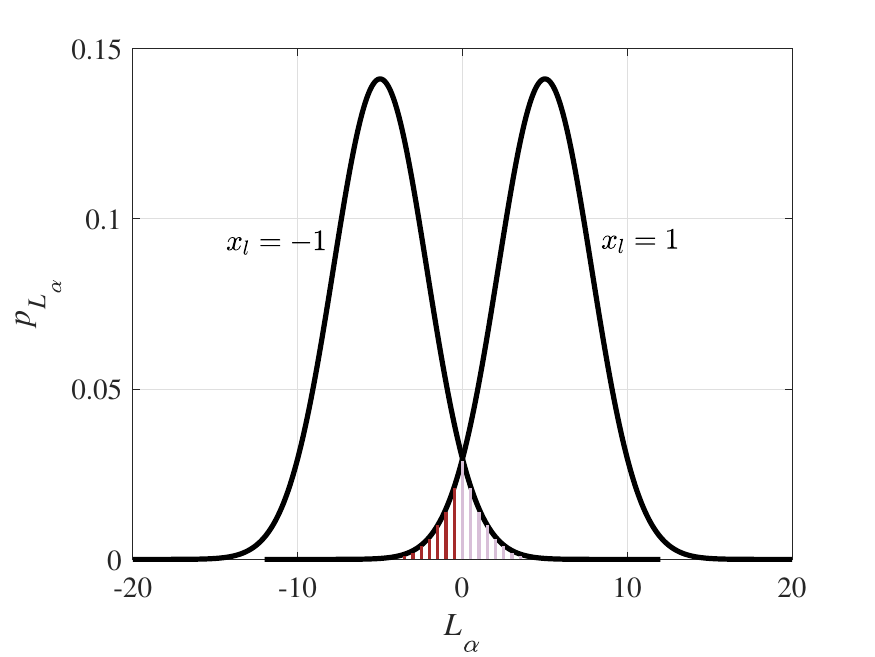}
		\caption{The probability of wrong-decision $P_i^l({\theta _j} \to W)$ represented by the integral of ${p_\psi }(\psi _i^l|{x_l=\mp 1})$ (see the shadowing parts).}\label{fig:PILW}
	\end{figure}
	
	\textup{Combined with {\textbf{\textit{Theorem} \ref{the:1}}} , we define the probability of right-decision $P_i^l({\theta _j} \to R)$ and wrong-decision $P_i^l({\theta _j} \to W)$ at VNs. The corresponding shadowing parts in Fig. \ref{fig:PILR} and Fig. \ref{fig:PILW} are presented to better illustrate $P_i^l({\theta _j} \to R)$ and $P_i^l({\theta _j} \to W)$, respectively. Then The calculation of $P_i^l({\theta _j} \to R)$ and $P_i^l({\theta _j} \to W)$ are derived as (\ref{equ:pilR}) and (\ref{equ:pilW}). It is noted that the detailed deduction of (\ref{equ:pilR}.a) and (\ref{equ:pilW}.a) is shown in Appendix B.}
	\begin{figure*}[h]
		\begin{equation}\label{equ:pilR}
		\begin{split}
		P_i^l({\theta _j} \to R) =& \sum\nolimits_{{x_l} = {\theta _j}} {P({x_l})P({\theta _j}|{x_l})} \\
		=& \frac{1}{2}\int_0^\infty  {{p_\psi }(\psi _i^l|{x_l} = {\theta _j} = 1 )}d\psi _i^l  + \frac{1}{2}\int_{ - \infty }^0 {{p_\psi }(\psi _i^l|{x_l} = {\theta _j} =  - 1 )}d\psi _i^l\\
		\overset{(a)}{=}& \frac{1}{2{\sqrt {2\pi \mathit{var_{\psi_{i}^l}}} }}\left\lbrace \int_0^\infty{\exp \left( { - \frac{{{{\left( {\psi _i^l - \mathit{var_{\psi_{i}^l}}/2} \right)}^2}}}{{2\mathit{var_{\psi_{i}^l}}}}} \right)}d\psi _i^l+\int_{ - \infty }^0{\exp \left( { - \frac{{{{\left( {\psi _i^l +  \mathit{var_{\psi_{i}^l}}/2} \right)}^2}}}{{2\mathit{var_{\psi_{i}^l}}}}} \right)}d\psi _i^l\right\rbrace\\
		\overset{(b)}{=}& \frac{1}{4}\left[1 + {\mathit{erf}\left( {\sqrt {\frac{\mathit{var_{\psi_{i}^l}}}{8}} } \right)} \right] + \frac{1}{4}\left[1 + {\mathit{erf}\left( {\sqrt {\frac{\mathit{var_{\psi_{i}^l}}}{8}} } \right)} \right]\\
		=& \frac{1}{2}\left[ {1+\mathit{erf}\left( {\sqrt {\frac{\mathit{var_{\psi_{i}^l}}}{8}} } \right)}\right].
		\end{split}
		\end{equation}
	\end{figure*}
	\begin{figure*}[h]
		\begin{equation}\label{equ:pilW}
		\begin{split}
		P_i^l({\theta _j} \to W) =& \sum\nolimits_{{x_l} \ne {\theta _j}} {P({x_l})P({\theta _j}|{x_l})} \\
		=& \frac{1}{2}\int_0^\infty  {{p_\psi }(\psi _i^l|{x_l} =  - {\theta _j} =  - 1 )d\psi _i^l}  + \frac{1}{2}\int_{ - \infty }^0 {{p_\psi }(\psi _i^l|{x_l} =  - {\theta _j} = 1 )d\psi _i^l}\\
		\overset{(a)}{=}& \frac{1}{2{\sqrt {2\pi \mathit{var_{\psi_{i}^l}}} }}\left\lbrace \int_0^\infty{\exp \left( { - \frac{{{{\left( {\psi _i^l + \mathit{var_{\psi_{i}^l}}/2} \right)}^2}}}{{2\mathit{var_{\psi_{i}^l}}}}} \right)d\psi _i^l}+\int_{ - \infty }^0{\exp \left( { - \frac{{{{\left( {\psi _i^l -  \mathit{var_{\psi_{i}^l}}/2} \right)}^2}}}{{2\mathit{var_{\psi_{i}^l}}}}} \right)d\psi _i^l}\right\rbrace\\
		\overset{(b)}{=}& \frac{1}{4}\left[ {\mathit{erfc}\left( {\sqrt {\frac{\mathit{var_{\psi_{i}^l}}}{8}} } \right)} \right] + \frac{1}{4}\left[ {\mathit{erfc}\left( {\sqrt {\frac{\mathit{var_{\psi_{i}^l}}}{8}} } \right)} \right]\\
		=& \frac{1}{2}{\mathit{erfc}\left( {\sqrt {\frac{\mathit{var_{\psi_{i}^l}}}{8}} } \right)}.
		\end{split}
		\end{equation}
		\hrulefill
	\end{figure*}
	
	\textup{Hence, we can conclude the probability $P_i^{l}(\theta_j)$ estimated at VN $v_{l}$ as (\ref{equ:PilP}) and (\ref{equ:PilW}), which discuss the probability $P_i^{l}(\theta_j)$ based on $P_i^l({\theta _j} \to R)$ and $P_i^l({\theta _j} \to W)$, as well as the uncertain transmit symbol $x_l=\mp 1$.}
	\QEDA
\end{proof}

{\textbf{\textit{Theorem} \ref{the:2}}} presents the expression of probability $P_i^{l}(\theta_j)$, based on which we can further deduct the ECV $\mathit{var_{\omega_i^l}}$ calculated at AIC.

\begin{theorem}\label{the:3}
	\textup{According to (\ref{equ:ECVori}) in {\textbf{\textit{Lemma} \ref{lem:1}}}, the ECV $\mathit{var_{\omega_i^l}}$ calculated at AIC can be calculated as}
	\begin{equation}\label{equ:ECVfed}
	\mathit{var_{\omega_i^l}} = \frac{{2h_{i,l}^2}}{{\sum\nolimits_{l' = 1,l' \ne l}^{{N_T}} {\left| {h_{i,l'}}\right|^2 \mathbb{V}(x_{l'})+\sigma _n^2}}},
	\end{equation}
	\textup{where}
	\begin{equation}\label{equ:Vxl}
	\mathbb{V}({x_l}) = \left\{ \!{\begin{array}{*{20}{l}}
		\!{1,{\:}t = 1},\\
		\!{\!\left[\!1\!+\!\mathit{erf\!}\!\left(\! {\sqrt {\frac{{\mathit{var\!_{\psi_i^l}}}}{8}} }\! \right)\right]\!\! \mathit{erfc\!}\!\left(\! {\sqrt {\frac{{\mathit{var\!_{\psi_i^l}}}}{8}} }\! \right),{\:}t\!>\!1},
		\end{array}} \right.
	\end{equation}
	\textup{in which the first line is calculated according to the first-iteration variance $\mathbb{V}(x_{l})$ in (\ref{equ:sigit1}), and $t$ denotes the iteration times during the MIMO detection.}
\end{theorem}
\begin{proof}
	\textup{The variance $\mathbb{V}({x_l})$ in the first iteration is calculated as in (\ref{equ:sigit1}), thus we mainly derive $\mathbb{V}({x_l})$ when $t>1$ in this theorem.}
	
	\textup{based on {\textbf{\textit{Theorem} \ref{the:2}}}  and {\textbf{\textit{Lemma} \ref{lem:1}}} we discuss the ECV $\mathit{var_{\omega_i^l}}$ calculated at AIC on two cases, i.e., the transmit symbol $x_l = 1$ and $x_l = -1$. The specific derivation is shown as follows:}
	\begin{itemize}
		\item[\textit{C1:}] \textup{The transmit symbol $x_l = 1$. With the probabilities $P_i^l(+1)$ estimated at VN $v_l$ shown as (\ref{equ:PilP}) and (\ref{equ:PilW}), we can obtain the mean value $\mathbb{E}(x_l^+)$ of $x_l$ according to (\ref{equ:Exl}) as}
		\begin{equation}
		\begin{split}
		\mathbb{E}({x_{l^+}} ) =& \frac{+1}{{\!2 }}\! \left[\!1\!+\!\mathit{erf}\!\left(\! {\sqrt {\frac{{\mathit{var_{\psi_i^l}}}}{8}} } \right)\!\right]  \!-\! \frac{1}{{\!2 }} \mathit{erfc}\!\left(\! {\sqrt {\frac{{\mathit{var_{\psi_i^l}}}}{8}} } \right)\\
		=& \mathit{erf}\left( {\sqrt {\frac{{\mathit{var_{\psi_i^l}}}}{8}} } \right) .
		\end{split}
		\end{equation}
		
		\textup{Then based on (\ref{equ:Varxl}), the variance $\mathbb{V}(x_l^+)$ of $x_l$ can be given as}
		\begin{equation}\label{equ:Vpo}
		\begin{split}
		\mathbb{V}(x_l^+\!) =& \left( {1}\right) ^2{P_i^{l}}({1})\!+\!\left( {-1}\right) ^2{P_i^{l}}({-1})\!-\!\left| {\mathbb{E}(x_{l})}\right|^2\\
		=& \!\left[\! {P_i^{l}}({1})\!+\!{P_i^{l}}({-1})\!\right] \!-\! \left[\! {\mathit{erf}\left( \!{\sqrt {\frac{{\mathit{var_{\psi_i^l}}}}{8}} } \right)\!} \!\right]^2\\
		=& \! {1-\! {\mathit{erf}\left( \!{\sqrt {\frac{{\mathit{var_{\psi_i^l}}}}{8}} } \right)\!}^2} \\
		=& \!\left[ \!{1\!+\! {\mathit{erf}\left( \!{\sqrt {\frac{{\mathit{var_{\psi_i^l}}}}{8}} } \right)} }\right]\left[ \!{1\!-\! {\mathit{erf}\left( \!{\sqrt {\frac{{\mathit{var_{\psi_i^l}}}}{8}} } \right)\!}}\right]\\
		=& \!\left[ \!{1\!+\! {\mathit{erf}\left( \!{\sqrt {\frac{{\mathit{var_{\psi_i^l}}}}{8}} } \right)\!} \!}\right]{\mathit{erfc}\left( \!{\sqrt {\frac{{\mathit{var_{\psi_i^l}}}}{8}} } \right)\!}.
		\end{split}
		\end{equation}
		\item[\textit{C2:}] \textup{The transmit symbol $x_l = -1$. With the probabilities $P_i^l(\mp 1)$ estimated at VN $v_l$ shown as (\ref{equ:PilP}) and (\ref{equ:PilW}), we can obtain the mean value $\mathbb{E}(x_l^-)$ of $x_l$ according to (\ref{equ:Exl}) as}
		\begin{equation}
		\begin{split}
		\mathbb{E}({x_{l^-}} ) =& \frac{-\!1}{{2 }}\!\left[\!1\!+\!\mathit{erf\!}\left(\! {\sqrt {\frac{{\mathit{var\!_{\psi_i^l}}}}{8}} } \right)\!\right]\!\!+\!\!\frac{1}{{2 }}\!\mathit{erfc\!}\left(\! {\sqrt {\frac{{\mathit{var\!_{\psi_i^l}}}}{8}} } \right)\\
		=& \frac{1}{{2 }}\left[{\mathit{erfc}\left(\!{\sqrt {\frac{{\mathit{var_{\psi_i^l}}}}{8}} } \right)\!-\!\mathit{erf}\left(\!{\sqrt {\frac{{\mathit{var_{\psi_i^l}}}}{8}} } \right)\!-\!1} \right]\\
		=& -\mathit{erf}\left( {\sqrt {\frac{{\mathit{var_{\psi_i^l}}}}{8}} } \right) .
		\end{split}
		\end{equation}
		
		\textup{Similarly, the variance $\mathbb{V}(x_l^-)$ of $x_l$ is expressed as}
		\begin{equation}
		\begin{split}
		\mathbb{V}(x_l^-) =&1 \!-\! \left[ {\mathit{erf}\left( {\sqrt {\frac{{\mathit{var_{\psi_i^l}}}}{8}} } \right)} \right]^2\\
		=& \!\left[ \!{1\!+\! {\mathit{erf}\left( \!{\sqrt {\frac{{\mathit{var_{\psi_i^l}}}}{8}} } \right)\!} }\right]{\mathit{erfc}\left( \!{\sqrt {\frac{{\mathit{var_{\psi_i^l}}}}{8}} } \right)\!}.
		\end{split}
		\end{equation}
		\textup{which is identical to (\ref{equ:Vpo}).}
	\end{itemize}
	
	\textup{Therefore, we can conclude that the variance $\mathbb{V}(x_l)$ is unrelated to the value of transmit symbol $x_l$, and can be represented identically as (\ref{equ:Vxl}). Then substitue (\ref{equ:Vxl}) into (\ref{equ:varG}) and (\ref{equ:ECVori}), we can acquire this {{{theorem}}}.}
	\QEDA
\end{proof}

Consequently, given the variance $\mathit{var_{\omega_i^l}}$ at AIC as (\ref{equ:ECVfed}) and (\ref{equ:Vxl}) in {\textbf{\textit{Theorem} \ref{the:3}}}, the mutual information $I_{\omega_{i}^l}$ can be calculated according to (\ref{equ:Iomega}) at the EIC (i.e., ON sub-detector). Then the mutual information $I_{\psi_i^l}$ at the VN sub-detector can be computed as (\ref{equ:Ipsi}). In addition, the mutual information $I_{L_l}$ between the LLR $L_l$ output by the VN sub-detector and transmit symbol $x_l$ is given as
\begin{equation}\label{equ:ILl}
I_{L_l} \buildrel \Delta \over = \mathit{F_{L}}\left({{\mathit{var_{\omega_i^l}}}}\right)= J\left( {\sqrt{\sum_{i'=1}^{N_R}{\mathit{var_{\omega_{i'}^l}}}}} \right),
\end{equation} 
then the averaged mutual information (AMI) of the VNs $I_L$ can be expressed as
\begin{equation}\label{equ:ILA}
I_{L} = \frac{1}{N_T}\sum_{l=1}^{N_T}{I_{L_l}}.
\end{equation} 

To sum up, the error functions-aided analysis (EF-AA) mechanism for FG-based iterative MIMO detections under the BPSK modulation can be summarized as Algorithm \ref{alg:EXIT}.
\renewcommand{\algorithmicrequire}{\textbf{Input:}}
\renewcommand{\algorithmicensure}{\textbf{Output:}}
\begin{algorithm}
	\caption{The EF-AA mechanism for iterative MIMO detections (BPSK modulation).}\label{alg:EXIT}
	\begin{algorithmic}[1] 
		\Require the channel matrix ${\boldsymbol{H}}$ estimated at receiver, the variance (power) of noise $\sigma_n^2$, the maximum iteration number $T$ for MIMO detections.
		\State {Initialize the counter for the iterative MIMO detection as $t=1$.}
		\Repeat
		\For {$l=1$ to $N_T$}
		\For {$i=1$ to $N_R$}
		\State{Calculate the variance $\mathit{var_{\omega_i^l}}$ according to (\ref{equ:ECVfed}) and (\ref{equ:Vxl}) at AIC, then transfer it to the EIC.}
		\State{Update the mutual information $I_{\omega_{i}^l}$ as (\ref{equ:Iomega}) at EIC.}
		\State{Update the variance $\mathit{var_{\psi_i^l}}$ as (\ref{equ:varpsi}) at EIC, then transfer it to the VN sub-detector.}
		\State{Compute the mutual information $I_{\psi_i^l}$ based on (\ref{equ:Ipsi}) at the VN sub-detector.}
		\EndFor
		\State{Calculate the mutual information $I_{L_l}$ at the VN sub-detector according to (\ref{equ:ILl}).}
		\EndFor
		\State{Calculate the averaged output mutual information $I_{L}$ as (\ref{equ:ILA}).}
		\State{Update counter for the iterative MIMO detection as $t=t+1$.}
		\Until{The maximum iteration number $T$ is reached.}
		\Ensure The averaged output mutual information $I_{L}$ of VNs in each iteration.
	\end{algorithmic}
\end{algorithm}
\subsection{The Error Functions-Aided Analysis under QPSK Modulation}
According to \cite{23:MIMOR}, the complex-domain Rayleigh fading MIMO channel (\ref{equ:HCOM}) can be converted into an equivalent real-domain form, which extracts the real and imaginary part of each entry in (\ref{equ:HCOM}), respectively. The conversion can be presented by
\begin{equation}\label{equ:HR}
{{\boldsymbol{H}}_R} = \left[ {\begin{array}{*{20}{c}}
	{\Re ({\boldsymbol{H}})}&{ - \Im ({\boldsymbol{H}})}\\
	{\Im ({\boldsymbol{H}})}&{ - \Re ({\boldsymbol{H}})}
	\end{array}} \right],
\end{equation}
where $\Re\left(\cdot \right) $ and $\Im\left(\cdot \right)$ denote the real part and imaginary part of elements, respectively. Then the corresponding vectors in (\ref{equ:HCOM}) can be transformed by
\begin{equation}\label{equ:yR}
{{\boldsymbol{y}}_R} = \left[ {\begin{array}{*{20}{c}}
	{\Re ({\boldsymbol{y}})}\\
	{\Im ({\boldsymbol{y}})}
	\end{array}} \right],
\end{equation}
\begin{equation}\label{equ:xR}
{{\boldsymbol{x}}_R} = \left[ {\begin{array}{*{20}{c}}
	{\Re ({\boldsymbol{x}})}\\
	{\Im ({\boldsymbol{x}})}
	\end{array}} \right],
\end{equation}
and
\begin{equation}\label{equ:nR}
{{\boldsymbol{n}}_R} = \left[ {\begin{array}{*{20}{c}}
	{\Re ({\boldsymbol{n}})}\\
	{\Im ({\boldsymbol{n}})}
	\end{array}} \right],
\end{equation}  
where each element corresponds to the information bit before QPSK modulation.

Therefore, through the aforementioned real-domain conversion, this paper extends the BPSK-modulated EF-AA mechanism to QPSK-modulated scenarios. In the following, the subscripts of matrix $\boldsymbol{H}_R$ and vectors $\boldsymbol{x}_R$, $\boldsymbol{y}_R$, $\boldsymbol{n}_R$ is omitted to better illustrate the derivation process.

Specifically, the complex QPSK symbol $x_l = \theta \in {1}/{\sqrt 2}\left\lbrace {1+1i,1-1i,-1+1i,-1-1i}\right\rbrace $ can be transfomred into the real form $x_{l}= \theta_R \in {1}/{\sqrt 2}\left\lbrace {1,-1}\right\rbrace$. Then the Gaussian approximation of the received symbol $y_{i}$ is revised as
\begin{equation}\label{equ:GAR}
\begin{split}
{y_{i}} &= \sum\nolimits_{l = 1}^{2N_T} {{h_{i,l}}{x_{l}}}  + {n_{i}}\\
&= {h_{i,l}}{x_{l}} + \underbrace {\sum\nolimits_{l' = 1,l' \ne l}^{2N_T} {{h_{i,l'}}{x_{l'}}}  + {n_{i}}}_{{g_{il}}},
\end{split}
\end{equation}
correspondingly, the variance $\mathbb{V}(x_{l})$ in the first iteration can be given as
\begin{equation}\label{equ:sigit1R}
\begin{split}
\mathbb{V}(x_{l})=& \frac{1}{2}\!\left[\!{{\left( \frac{1}{\sqrt{2}}\right)}^2\!+\!{\left( \frac{-1}{\sqrt{2}}\right)}^2}\!\right]\!-\!\left| \frac{1}{2}\left( {\frac{1}{\sqrt{2}}\!-\!\frac{1}{\sqrt{2}}}\right) \right|^2\\
=& \frac{1}{2}.
\end{split} 
\end{equation}

Under the QPSK modulation, the ECV $\mathit{var_{\omega_i^l}}$ defined in {\textbf{\textit{Lemma} \ref{lem:1}}} and variance $\mathit{var_{\psi_i^l}}$ in {\textbf{\textit{Lemma} \ref{lem:2}}} are computed as
\begin{equation}\label{equ:ECVoriR}
\mathit{var_{\omega_i^l}} = \frac{{2h_{i,l}^{2}}}{{\sigma _{{g_{il}}}^{2}}},
\end{equation}
and
\begin{equation}\label{equ:varpsiR}
\mathit{var_{\psi_i^l}}=\sum\nolimits_{i'=1,i'\ne i}^{2N_R}{J^{-1}\left({I_{\omega_i^{l'}}}\right)},
\end{equation}
respectively.

In addition, (\ref{equ:PilP}) and (\ref{equ:PilW}) in {\textbf{\textit{Theorem} \ref{the:2}}} need to be modified as
\begin{equation}\label{equ:PilPR}
\begin{split}
P_i^l( + 1/\sqrt 2 ) =& P_i^l(\theta_j =  + 1/\sqrt 2 )\\
=& \left\{ {\begin{array}{*{20}{l}}
	{\frac{1}{2}\left[ 1+\mathit{erf}\left( {\sqrt {{\mathit{var_{\psi_i^l}}}/{8}} } \right)\right] ,if\:{x_l} =  + 1/\sqrt 2 ,}\\
	{\frac{1}{2}\mathit{erfc}\left( {\sqrt {{\mathit{var_{\psi_i^l}}}/{8}} } \right),if\:{x_l} =  - 1/\sqrt 2,}
	\end{array}} \right.
\end{split}
\end{equation}
\textup{and}
\begin{equation}\label{equ:PilWR}
\begin{split}
P_i^{l}(- 1/\sqrt 2) =& P_i^l(\theta_j =  - 1/\sqrt 2 )\\
=& \left\{ {\begin{array}{*{20}{l}}
	{\frac{1}{2}\left[ 1+\mathit{erf}\left( {\sqrt {{\mathit{var_{\psi_i^l}}}/{8}} } \right)\right],if\:{x_l} =  - 1/\sqrt 2 ,}\\
	{\frac{1}{2}\mathit{erfc}\left( {\sqrt {{\mathit{var_{\psi_i^l}}}/{8}} } \right),if\:{x_l} =  + 1/\sqrt 2,}
	\end{array}} \right.
\end{split}
\end{equation}

Therefore, the ECV $\mathit{var_{\omega_i^l}}$ under the QPSK-modulated scenario can be calculated as in {\textbf{\textit{Theorem} \ref{the:4}}}.
\begin{theorem}\label{the:4}
	\textup{Under the QPSK-modulated scenario, the ECV $\mathit{var_{\omega_i^l}}$ calculated at AIC is given as}
	\begin{equation}\label{equ:ECVfedR}
	\mathit{var_{\omega_i^l}} = \frac{{4h_{i,l}^2}}{{\!\sum\limits_{l' = 1,l' \ne l}^{{2N_T}}\! {\left| {h_{i,l'}}\right|^2 \!{\left[\!1\!+\!\mathit{erf\!}\!\left(\! {\sqrt {\frac{{\mathit{var\!_{\psi_i^l}}}}{8}} }\! \right)\right]\!\! \mathit{erfc\!}\!\left(\! {\sqrt {\frac{{\mathit{var\!_{\psi_i^l}}}}{8}} }\! \right)\!+\!2\sigma _n^2}}}}.
	\end{equation}
\end{theorem}
\begin{proof}
	\textup{Similarly as in {\textbf{\textit{Theorem} \ref{the:3}}}, under the QPSK modulation, the ECV $\mathit{var_{\omega_i^l}}$ in the subsequent iterations at AIC can be calculated depending on the following two cases:}
	\begin{itemize}
		\item[\textit{C1:}] \textup{The transmit symbol $x_l = 1/\sqrt{2}$. With the probabilities $P_i^l(\mp 1/\sqrt{2})$ estimated at VN $v_l$ shown as (\ref{equ:PilPR}) and (\ref{equ:PilWR}), the mean value $\mathbb{E}(x_l^+)$ and variance $\mathbb{V}(x_l^+)$ of $x_l$ can be expressed as}
		\begin{equation}
		\begin{split}
		\mathbb{E}({x_{l^+}} ) =& \frac{+1}{{\!2\sqrt 2 }}\! \left[\!1\!+\!\mathit{erf}\!\left(\! {\sqrt {\frac{{\mathit{var_{\psi_i^l}}}}{8}} } \right)\!\right]  \!-\! \frac{1}{{\!2\sqrt 2 }} \mathit{erfc}\!\left(\! {\sqrt {\frac{{\mathit{var_{\psi_i^l}}}}{8}} } \right)\\
		=& \frac{1}{{\sqrt 2 }}\mathit{erf}\!\left( {\sqrt {\frac{{\mathit{var_{\psi_i^l}}}}{8}} } \right),
		\end{split}
		\end{equation}
		\textup{and}
		\begin{equation}\label{equ:VpoR}
		\begin{split}
		\mathbb{V}(x_l^+\!) =& \left( \frac{1}{\sqrt{2}}\right) ^2{P_i^{l}}(\frac{1}{\sqrt{2}})\!+\!\left( \frac{-1}{\sqrt{2}}\right) ^2{P_i^{l}}(\frac{-1}{\sqrt{2}})\!-\!\left| {\mathbb{E}(x_{l})}\right|^2\\
		=& \frac{1}{2}\!\left[\! {P_i^{l}}(\frac{-1}{\sqrt{2}})\!+\!{P_i^{l}}(\frac{1}{\sqrt{2}})\!\right] \!-\! \frac{1}{{ 2 }}\left[\! {\mathit{erf}\left( \!{\sqrt {\frac{{\mathit{var_{\psi_i^l}}}}{8}} } \right)\!} \!\right]^2\\
		=& \frac{1}{2}\!\left[ {1-\! {\mathit{erf}\left( \!{\sqrt {\frac{{\mathit{var_{\psi_i^l}}}}{8}} } \right)\!}^2}\right] \\
		=& \frac{1}{2}\!\left[ \!{1\!+\! {\mathit{erf}\left( \!{\sqrt {\frac{{\mathit{var_{\psi_i^l}}}}{8}} } \right)\!} \!}\right]{\mathit{erfc}\left( \!{\sqrt {\frac{{\mathit{var_{\psi_i^l}}}}{8}} } \right)\!},
		\end{split}
		\end{equation}
		\textup{respectively.}
		\item[\textit{C2:}] \textup{The transmit symbol $x_l = -1/\sqrt{2}$. With the probabilities $P_i^l(\mp 1/\sqrt{2})$ estimated at VN $v_l$ shown as (\ref{equ:PilPR}) and (\ref{equ:PilWR}), we can obtain the mean value $\mathbb{E}(x_l^-)$ and variance $\mathbb{V}(x_l^-)$ of $x_l$ as}
		\begin{equation}
		\begin{split}
		\mathbb{E}({x_{l^-}} ) =& \frac{-\!1}{{2\sqrt 2 }}\!\left[\!1\!+\!\mathit{erf\!}\left(\! {\sqrt {\frac{{\mathit{var\!_{\psi_i^l}}}}{8}} } \right)\!\right]\!\!+\!\!\frac{1}{{2\sqrt 2 }}\!\mathit{erfc\!}\left(\! {\sqrt {\frac{{\mathit{var\!_{\psi_i^l}}}}{8}} } \right)\\
		=& -\frac{1}{{\sqrt 2 }}\mathit{erf}\left( {\sqrt {\frac{{\mathit{var_{\psi_i^l}}}}{8}} } \right),
		\end{split}
		\end{equation}
		\textup{and}
		\begin{equation}
		\begin{split}
		\mathbb{V}(x_l^-) =&\frac{1}{2} \!-\! \frac{1}{{2}}\left[ {\mathit{erf}\left( {\sqrt {\frac{{\mathit{var_{\psi_i^l}}}}{8}} } \right)} \right]^2\\
		=& \frac{1}{2}\!\left[ \!{1\!+\! {\mathit{erf}\left( \!{\sqrt {\frac{{\mathit{var_{\psi_i^l}}}}{8}} } \right)\!} }\right]{\mathit{erfc}\left( \!{\sqrt {\frac{{\mathit{var_{\psi_i^l}}}}{8}} } \right)\!},
		\end{split}
		\end{equation}
		\textup{respectively.}
	\end{itemize}
	
	\textup{It is obvious that under the QPSK modulation, the variance $\mathbb{V}(x_l)$ of $x_l$ in subsequent iterations is irrelevant to the specific sign of $x_l$. Additionally, through initializing the variance of extrinsic LLR $\psi_i^l$ as $\mathit{var_{\psi_{i}^l}}=0$, we can uniformly express the variance $\mathbb{V}(x_l)$ under different iterations as}
	\begin{equation}\label{equ:VXLR}
	\mathbb{V}(x_l) =\frac{1}{2}\!\left[ \!{1\!+\! {\mathit{erf}\left( \!{\sqrt {\frac{{\mathit{var_{\psi_i^l}}}}{8}} } \right)\!} }\right]{\mathit{erfc}\left( \!{\sqrt {\frac{{\mathit{var_{\psi_i^l}}}}{8}} } \right)\!}.
	\end{equation}
	
	\textup{Therefore, by substituting (\ref{equ:VXLR}) into (\ref{equ:ECVfed}), the variance of LLR $\mathit{var_{\omega_i^l}}$ can be presented by (\ref{equ:ECVfedR}).} 
	\QEDA
\end{proof}

Referring to (\ref{equ:varpsiR}) and {\textbf{\textit{Theorem} \ref{the:3}}}, the EF-AA mechanism under QPSK modulation is concluded in Algorithm \ref{alg:EXIT2}. 
\renewcommand{\algorithmicrequire}{\textbf{Input:}}
\renewcommand{\algorithmicensure}{\textbf{Output:}}
\begin{algorithm}
	\caption{The EF-AA mechanism for iterative MIMO detections (QPSK modulation).}\label{alg:EXIT2}
	\begin{algorithmic}[1] 
		\Require the channel matrix ${\boldsymbol{H}}$ estimated at receiver, the variance (power) of noise $\sigma_n^2$, the maximum iteration number $T$ for MIMO detections.
		\State {Initialize the counter for the iterative MIMO detection as $t=1$, and the variance $\mathit{var_{\psi_i^l}}=0$.}
		\Repeat
		\For {$l=1$ to $2N_T$}
		\For {$i=1$ to $2N_R$}
		\State{Calculate the variance $\mathit{var_{\omega_i^l}}$ transferred from AIC to EIC according to (\ref{equ:ECVfedR}).}
		\State{Update the mutual information $I_{\omega_{i}^l}$ as (\ref{equ:Iomega}).}
		\State{Update the variance $\mathit{var_{\psi_i^l}}$ transferred from EIC to the VN sub-detector based on (\ref{equ:varpsiR}).}
		\State{Compute the mutual information $I_{\psi_i^l}$ as (\ref{equ:Ipsi}).}
		\EndFor
		\State{Calculate the mutual information $I_{L_l}$ as (\ref{equ:ILl}).}
		\EndFor
		\State{Calculate the averaged output mutual information $I_{L}$ as (\ref{equ:ILA}).}
		\State{Counter updation $t=t+1$.}
		\Until{The maximum iteration number $T$ is reached.}
		\Ensure The averaged output mutual information $I_{L}$ of VNs in each iteration.
	\end{algorithmic}
\end{algorithm}

\section{Numerical Results}
In this section, we utilize the proposed EF-AA mechanism to analyze the performance of FG-based iterative MIMO detections. Specifically, this section assesses the precision of the EF-AA mechanism from two aspects, i.e., the convergence under different iterations (by plotting AMI versus $I_{ter}$ in Fig. \ref{fig:MIIte}) and the convergence under different received SNRs (by plotting AMI versus $\rho_r$ in Fig. \ref{fig:MISNR}), respectively. More specifically, to verify the effectiveness of the proposed mutual information analysis method, the BER performance of the FG-based iterative MIMO detection under different iteration times and different received SNRs are presented in Fig. \ref{fig:BERIte} and Fig. \ref{fig:BERSNR}, respectively. In addition, the antennas at transmit and receive sides are equipped as $N_T=N_R=4,16$ for MIMO scenarios and $N_T=N_R=128,256$ for m-MIMO scenarios. In this section, the FG-BP iterative MIMO detection in \cite{8:FGEP} and QPSK modulation are adopted.

\begin{figure}[ht]
	\begin{center}
		\subfigure[{ $N_T=N_R=4,16.$}]{
			\includegraphics[width=1.76in,height=1.45in]{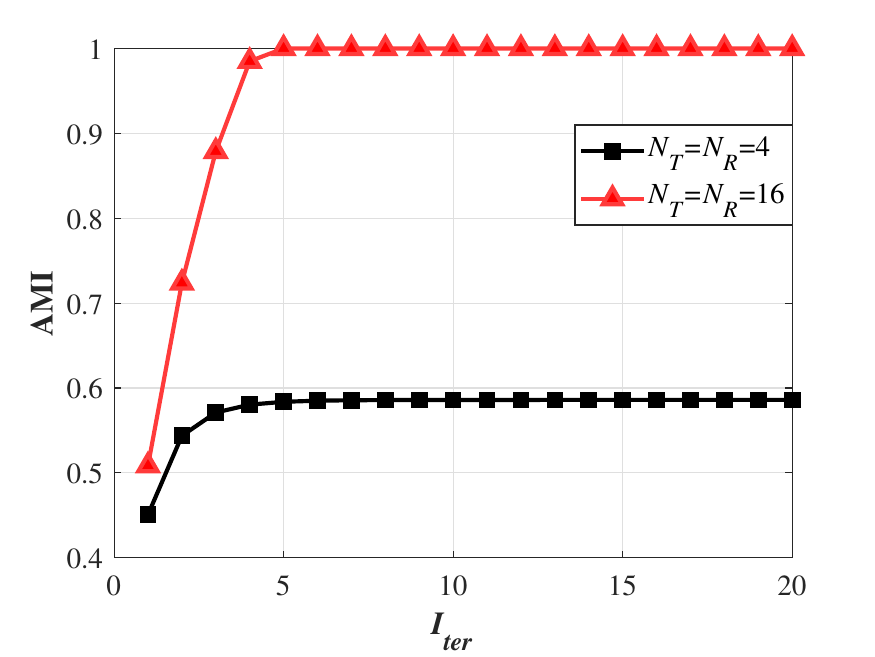}\label{fig:MIIte_a}
		}\hspace{-7.3mm}
		\subfigure[{ $N_T=N_R=128,256.$}]{
			\includegraphics[width=1.76in,height=1.45in]{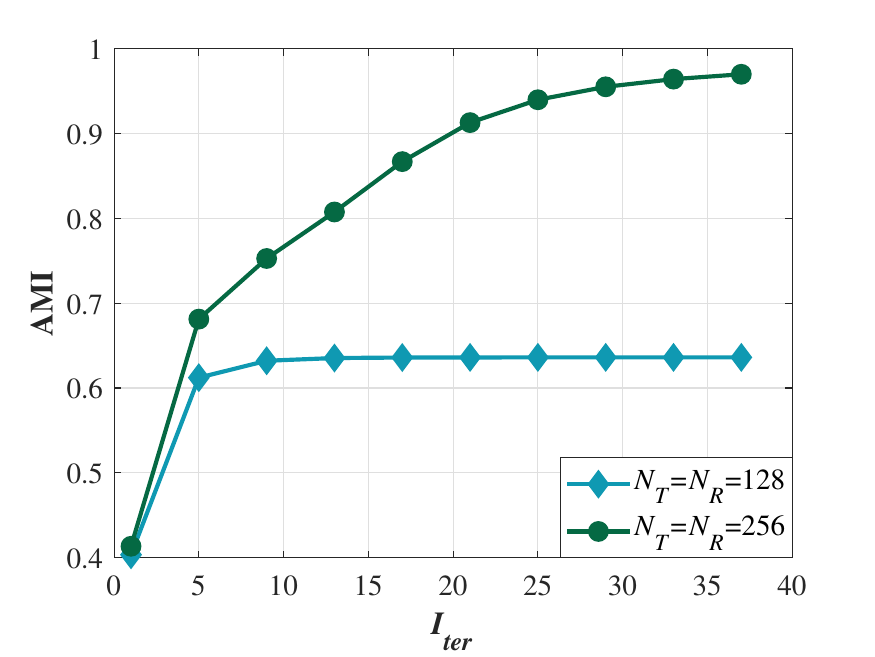}\label{fig:MIIte_b}
		}
		\caption{The AMI of the FG-based iterative MIMO detection under different iteration times for various antenna equipment (where $\rho_r=24$ dB and $\rho_r=4$ dB from left to right).}\label{fig:MIIte}
	\end{center}
\end{figure}
Fig. \ref{fig:MIIte} and Fig. \ref{fig:BERIte} show that the proposed EF-AA mechanism can evaluate the convergence of FG-based iterative MIMO detections accurately. For example, Fig. \ref{fig:MIIte_a} reveals that when the iteration number reaches 5, the AMI for both $N_T=N_R=4$ and $N_T=N_R=16$ approaches 1 identically. Moreover, the convergence performance of $N_T=N_R=4$ and $N_T=N_R=16$ in Fig. \ref{fig:MIIte_a} are consistent with the BER performance in Fig. \ref{fig:BERIte_a}, which tend to be steady after 5th iteration. Likewise, Fig. \ref{fig:MIIte_b} and Fig. \ref{fig:BERIte_b} present identical convergence characteristics in AMI and BER performance.
\begin{figure}[ht]
	\begin{center}
		\subfigure[$N_T=N_R=4,16.$]{
			\includegraphics[width=1.76in,height=1.45in]{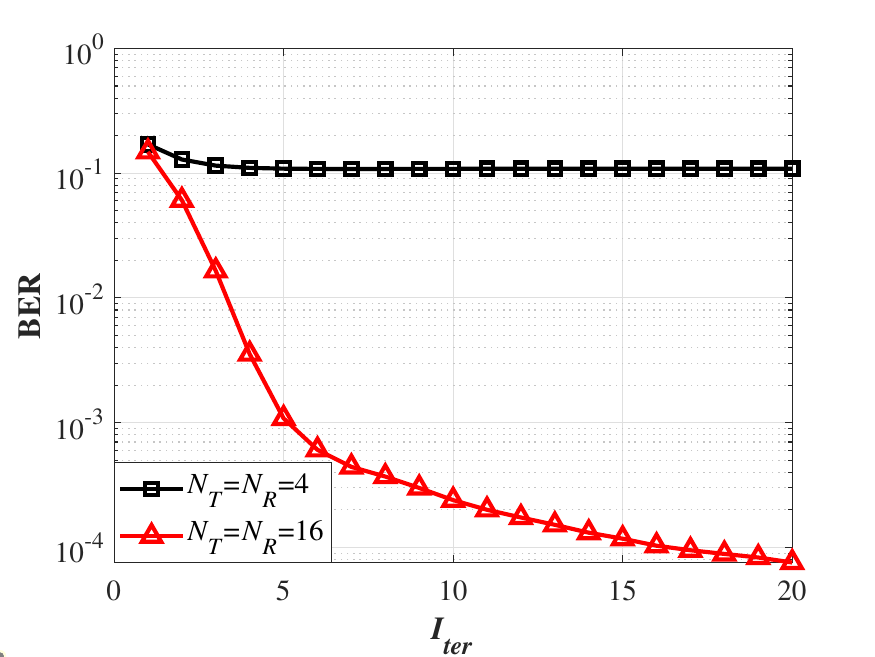}\label{fig:BERIte_a}
		}\hspace{-7.3mm}
		\subfigure[$N_T=N_R=128,256.$]{
			\includegraphics[width=1.76in,height=1.45in]{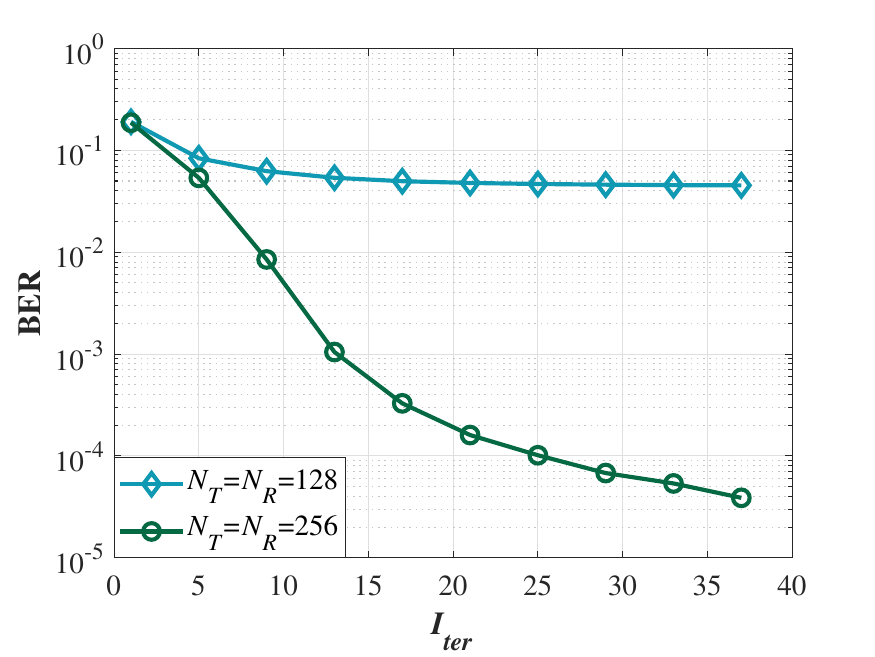}\label{fig:BERIte_b}
		}
		\caption{The BER performance of the FG-based iterative MIMO detection \cite{8:FGEP} under different iteration times for various antenna equipment (where $\rho_r=24$ dB and $\rho_r=4$ dB from left to right).}\label{fig:BERIte}
	\end{center}
\end{figure}

Additionally, Fig. \ref{fig:MISNRIte} is depicted to better demonstrate the evaluation accuracy of the proposed EF-AA mechanism on the convergence performance for the FG-based iterative MIMO detection. The left axis of Fig. \ref{fig:MISNRIte} exhibits the BER performance, and the right axis displays the characteristic of AMI. 
\begin{figure}
	\centering
	\includegraphics[width=0.5\textwidth]{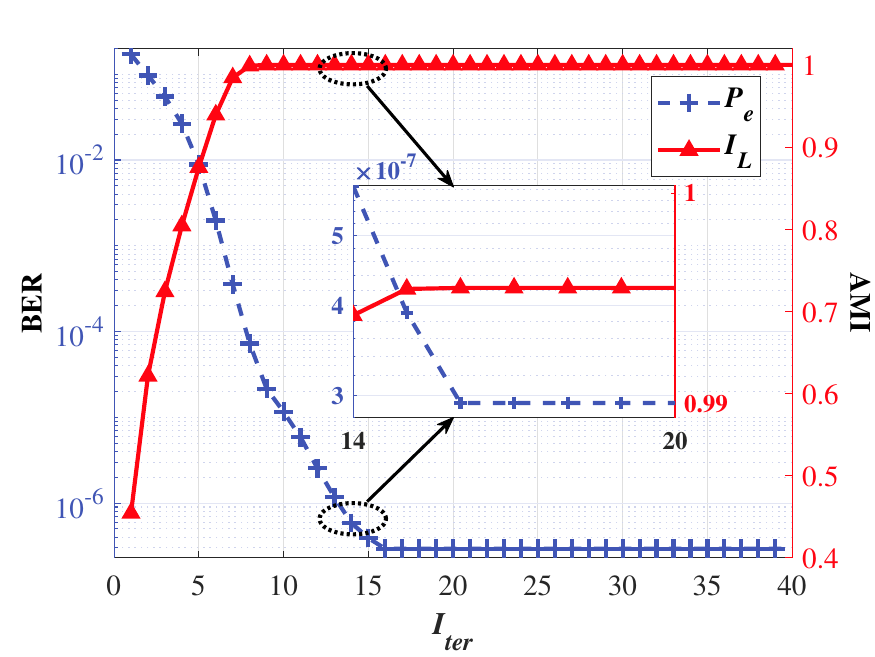}
	\caption{The specific comparison between the EF-AA mechanism-calculated AMI and BER performance of \cite{8:FGEP} (where $\rho_r=9$ dB).}\label{fig:MISNRIte}
\end{figure}

It can be seen from Fig. \ref{fig:MISNRIte} that with the increase of iteration number $I_{ter}$, the decline of $I_L$ (corresponding to AMI performance) and the growth of $P_e$ curve (corresponding to BER performance) are exactly homologous. Specifically, both the $I_L$ curve and $P_e$ curve tend to keep stable after the 16th iteration.

\begin{figure}
	\centering
	\includegraphics[width=0.5\textwidth]{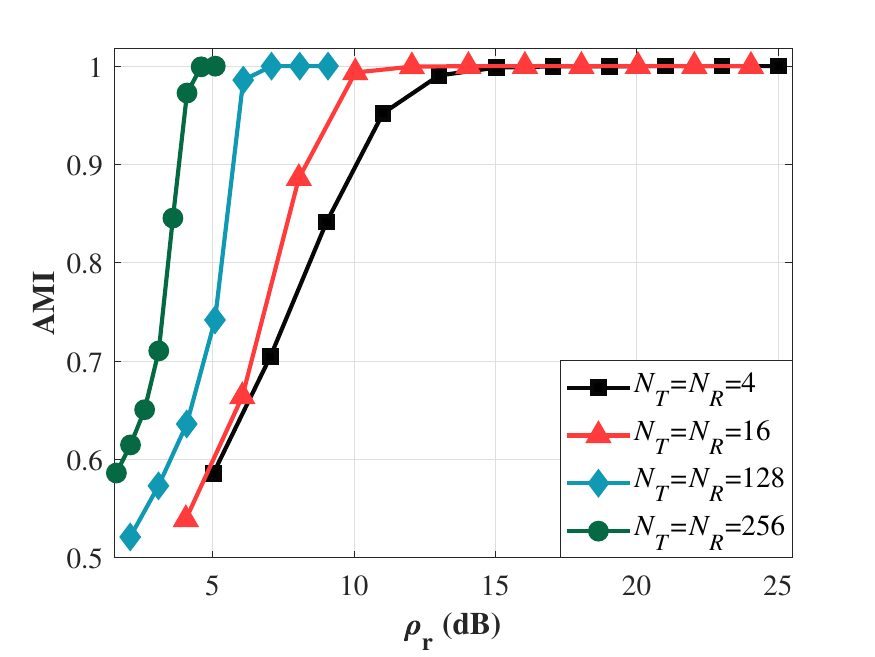}
	\caption{The EF-AA mechanism calculated-AMI versus received SNR $\rho_r$ of the FG-based iterative MIMO detection (where the maximum iteration number is $T=40$).}\label{fig:MISNR}
\end{figure}
\begin{figure}
	\centering
	\includegraphics[width=0.5\textwidth]{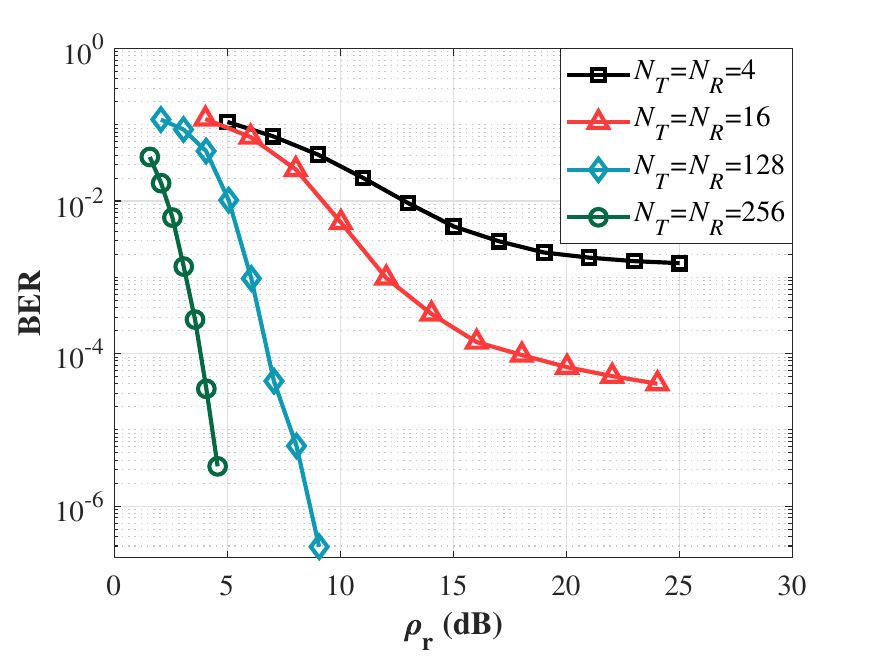}
	\caption{The BER performance of the FG-based iterative MIMO detection \cite{8:FGEP} versus received SNR $\rho_r$ (where the maximum iteration number is $T=40$).}\label{fig:BERSNR}
\end{figure}
Fig. \ref{fig:MISNR} and Fig. \ref{fig:BERSNR} compare the AMI computed by the EF-AA mechanism and the BER performance of the FG-based iterative MIMO detection \cite{8:FGEP} under different antenna types of equipment. It can be seen that as the growth of antennas, the AMI curve (in Fig. \ref{fig:MISNR}) rises faster, and the BER curve (in Fig. \ref{fig:BERSNR}) drops more dramatically. Additionally, Fig. \ref{fig:MISNR} and Fig. \ref{fig:BERSNR} also exhibit the effectiveness of the proposed EF-AA mechanism. From Fig. \ref{fig:MISNR}, the convergence SNR under different antennas can be concluded as $4.5,7,14,17$ dB (from left to right) respectively. Then Fig. \ref{fig:BERSNR} shows that under antenna $N_T=N_R=256,128$, when the received SNR $\rho_r$ reaches the convergence SNR in Fig. \ref{fig:MISNR}, the corresponding BERs tend to 0, which indicates the convergence of the FG-based iterative MIMO detection under m-MIMO scenarios. Then under antenna $N_T=N_R=16,4$, when the received SNR $\rho_r$ reaches the convergence SNR in Fig. \ref{fig:MISNR}, the corresponding BERs tend to be steady, which indicates the convergence of the FG-based iterative MIMO detection.

\section{Conclusions}
This paper investigated the mutual information update flow of the FG-based iterative MIMO detection, particularly focusing on the convergence performance under different iteration numbers and SNRs. In this paper, the EF-AA mechanism was proposed to provide exact mutual information curves through Gaussian approximation and closed-form calculation under BPSK and QPSK modulations. Numerical results of AMI and BER performance demonstrate that for the MIMO and m-MIMO scenarios, the proposed EF-AA mechanism can evaluate the convergence characteristic of FG-based iterative MIMO detections in both precise iteration times and received SNRs. It can be foreseen that the proposed EF-AA mechanism might have significant impacts both on performance evaluation and algorithm optimization.

\begin{appendices}
	\section{Curve Fitting Function $J\left(\cdot \right)$}
	As defined in \cite{24:MICAL}, the mutual information $I_\gamma$ between symbol $s$ and LLR $\gamma$ is given as
	\begin{equation}\label{equ:Iga}
	\begin{split}
	I_{\gamma} =& I({s};\gamma)\\
	=& \sum\limits_{j\!=\!1}^q\! {\int\limits_{ - \infty }^\infty \! {p\!\left( {s = {\theta _j}} \right)\!p\left( \!{\gamma |s = {\theta _j}} \right)\!{{\log }_2}\frac{{p\left( {\gamma |s = {\theta _j}} \right)}}{{\sum\limits_{j = 1}^q {p\left( {\gamma ,s = {\theta _j}} \right)} }}} }d\gamma,
	\end{split}
	\end{equation}  
	where $p\left( {s = {\theta _j}} \right)$ denotes the \textit{a prior} probability of the symbol $s$, $p\left( {\gamma | s} \right)$ denotes the conditional \textit{a posterior} probability of LLR $\gamma$, and $p\left( {\gamma , s} \right)$ denotes the joint \textit{a posterior} probability of LLR $\gamma$ and symbol $s$.
	
	Specifically, if the LLR $\gamma$ follows the Gaussian distribution with mean of ${\sigma ^2}/2$ and variance of ${\sigma ^2}$, then (\ref{equ:Iga}) can be expressed through the approximate curve fitting function $J\left(\cdot \right)$, which can be given as follows
	\begin{equation}
	J({\sigma})\!\approx\!\left\{\begin{array}{l}\!{a_{J,1}}\sigma^3\!+\!{b_{J,1}}\sigma^2\!+\!{e_{J,1}}\sigma^{},\quad\!0\!\le\!\sigma\!\le\!\sigma^*,\\ 
	\!1\!-\!{{e}^{({a_{J,2}}\sigma^3\!+\!{b_{J,2}}\sigma^2\!+\!{e_{J,2}}\sigma\!+\!{d_{J,2}})}},\!\sigma^*\!<\!\sigma\!<\!10,\\
	\!1,\quad\!\sigma\!\ge\!10,
	\end{array}\right.
	\end{equation}	
	\textup{in which}
	\begin{equation}
	\begin{aligned}
	\begin{array}{*{20}{l}}
	{\sigma^*\!=\!1.6363,\!}&{{a_{J,1}}\!=\!-\!0.0421061,\!}&{{b_{J,1}}\!=\!0.209252,}\\
	{{e_{J,1}}\!=\!-\!0.00640081,\!}&{{a_{J,2}}\!=\!0.00181491,\!}&{{b_{J,2}}\!=\!-\!0.142675,}\\
	{{e_{J,2}}\!=\!-\!0.0822054,\!}&{{d_{J,2}}\!=\!0.0549608.}&{}
	\end{array}
	\end{aligned}
	\end{equation}
	
	Furthermore, the inverse function of $J\left( \cdot\right)$ is represented as
	\begin{equation}
	{J^{-1}}(I_\sigma)\!\approx\!\left\{\begin{array}{l}
	\!{a_{\sigma ,1}}{I_\sigma^2}\!+\!{b_{\sigma ,1}}I_\sigma\!+\!{e_{\sigma,1}}\!\sqrt{\!I_\sigma},\quad 0\!\le\!I_\sigma\!\le\!{I^*_\sigma},\\ 
	\!-{a_{\sigma ,2}}\ln\left[{{b_{\sigma,2}}(1\!-\!I_\sigma)}\right]\!-\!{e_{\sigma,2}}I_\sigma,\quad {I^*_\sigma}\!<\!I_\sigma\!<\!1,
	\end{array}\right.
	\end{equation}
	where
	\begin{equation}
	\begin{aligned}
	\begin{array}{*{20}{l}}
	{\!I^*_\sigma\!=\!0.3646,\!}&{{a_{\sigma ,1}}\!=\!1.09542,\!}&{{b_{\sigma ,1}}\!=\!0.214217,}\\
	{\!{e_{\sigma ,1}}\!=\!2.33727,\!}&{{a_{\sigma ,2}}\!=\!0.706692,\!}&{{b_{\sigma ,2}}\!=\!0.386013,}\\
	{\!{e_{\sigma ,2}}\!=\!-\!1.75017.}&{}&{}
	\end{array}
	\end{aligned}
	\end{equation}	
	
	\section{Detailed Derivation of Right and Wrong Decision Probabilities}
	We show the details of the derivation of the equation of (\ref{equ:pilR}.b) and (\ref{equ:pilW}.b) in this appendix. 
	
	Firstly transform the first exponential term in equation (\ref{equ:pilR}.a) as $\frac{{{{\left( {\psi _i^l - \mathit{var_{\psi_{i}^l}}/2} \right)}^2}}}{{2\mathit{var_{\psi_{i}^l}}}} =  {t^2}$, thus we have $t = \frac{{\left( {\mathit{var_{\psi_{i}^l}}/2-\psi _i^l } \right)}}{{\sqrt {2\mathit{var_{\psi_{i}^l}}} }}$, then we have
	\begin{equation}\label{equ:R1}
	\begin{split}
	&\frac{1}{2{\sqrt {2\pi \mathit{var_{\psi_i^l}}} }}\int_0^\infty  {\exp \left( { - \frac{{{{\left( {\psi _i^l - \mathit{var_{\psi_i^l}}/2} \right)}^2}}}{{2\mathit{var_{\psi_i^l}}}}} \right)} d\psi _i^l\\
	&= \frac{1}{2{\sqrt {2\pi \mathit{var_{\psi_i^l}}} }} \times \sqrt {2\mathit{var_{\psi_i^l}}} \int_{ - \infty }^{\sqrt {{{\mathit{var_{\psi_i^l}}}}/{8}} } {\exp \left( { - {t^2}} \right)} dt\\
	&= \frac{1}{4}\left( {\frac{2}{{\sqrt \pi  }}\int_{0}^{\sqrt {{{\mathit{var_{\psi_i^l}}}}/{8}} }  {\exp \left( { - {t^2}} \right)} dt} +{\frac{2}{{\sqrt \pi  }}\int_{-\infty}^{0}  {\exp \left( { - {t^2}} \right)} dt}\right)\\
	&\overset{(c)}{=} \frac{1}{4}\left( {\frac{2}{{\sqrt \pi  }}\int_{0}^{\sqrt {{{\mathit{var_{\psi_i^l}}}}/{8}} }  {\exp \left( { - {t^2}} \right)} dt} +{\frac{2}{{\sqrt \pi  }}\int_{0}^{\infty}  {\exp \left( { - {t^2}} \right)} dt}\right)\\
	&= \frac{1}{4}\left[1 + \mathit{erf}\left( {\sqrt {\frac{{\mathit{var_{\psi_i^l}}}}{8}} } \right)\right],
	\end{split}
	\end{equation}
	where $d\psi _i^l = -\sqrt {2\mathit{var_{\psi_i^l}}} dt$, and the sub-equation (\ref{equ:R1}.c) holds because $exp(-t^2)$ is an even function. 
	
	Similarly, the second term of equation (\ref{equ:pilR}.a) can be given as
	\begin{equation}\label{equ:R2}
	\begin{split}
	&\frac{1}{2{\sqrt {2\pi \mathit{var_{\psi_i^l}}} }}\int_{ - \infty }^0 {\exp \left( { - \frac{{{{\left( {\psi _i^l + \mathit{var_{\psi_i^l}}/2} \right)}^2}}}{{2\mathit{var_{\psi_i^l}}}}} \right)} d\psi _i^l\\
	&= \frac{2}{4{\sqrt {2\pi \mathit{var_{\psi_i^l}}} }} \times \sqrt {2\mathit{var_{\psi_i^l}}} \int_{ - \infty }^{\sqrt {\mathit{var_{\psi_i^l}}/8} } {\exp \left( { - {t^2}} \right)} dt\\
	&= \frac{1}{4}\left[1 + \mathit{erf}\left( {\sqrt {\frac{{\mathit{var_{\psi_i^l}}}}{8}} } \right)\right],
	\end{split}
	\end{equation}
	where $t = \frac{{\left( {\psi _i^l+\mathit{var_{\psi_i^l}}/2} \right)}}{{\sqrt {2\mathit{var_{\psi_i^l}}} }}$, and $d\psi _i^l = \sqrt {2\mathit{var_{\psi_i^l}}} dt$.
	
	Follow the same principle of (\ref{equ:R1}) and (\ref{equ:R2}), we can deduct the terms of (\ref{equ:pilW}.a) as 
	\begin{equation}\label{equ:w1}
	\begin{split}
	&\frac{1}{2{\sqrt {2\pi \mathit{var_{\psi_i^l}}} }}\int_0^\infty  {\exp \left( { - \frac{{{{\left( {\psi _i^l + \mathit{var_{\psi_i^l}}/2} \right)}^2}}}{{2\mathit{var_{\psi_i^l}}}}} \right)} d\psi _i^l\\
	&= \frac{2}{4{\sqrt \pi  }}\int_{\sqrt {\frac{{\mathit{var_{\psi_i^l}}}}{8}} }^\infty  {\exp \left( { - {t^2}} \right)} dt\\
	&= \frac{1}{4}\mathit{erfc}\left( {\sqrt {\frac{{\mathit{var_{\psi_i^l}}}}{8}} } \right),
	\end{split}
	\end{equation}
	with $t = \frac{{\left( {\psi _i^l+\mathit{var_{\psi_i^l}}/2} \right)}}{{\sqrt {2\mathit{var_{\psi_i^l}}} }}$, and
	\begin{equation}\label{equ:w2}
	\begin{split}
	&\frac{1}{2{\sqrt {2\pi \mathit{var_{\psi_i^l}}} }}\int_{ - \infty }^0 {\exp \left( { - \frac{{{{\left( {\psi _i^l - \mathit{var_{\psi_i^l}}/2} \right)}^2}}}{{2\mathit{var_{\psi_i^l}}}}} \right)} d\psi _i^l\\
	&= \frac{2}{4{\sqrt \pi  }}\int_{\sqrt {\frac{{\mathit{var_{\psi_i^l}}}}{8}} }^\infty  {\exp \left( { - {t^2}} \right)} dt\\
	&= \frac{1}{4}\mathit{erfc}\left( {\sqrt {\frac{{\mathit{var_{\psi_i^l}}}}{8}} } \right),
	\end{split}
	\end{equation}
	with $t = \frac{{\left( {\mathit{var_{\psi_i^l}}/2-\psi _i^l} \right)}}{{\sqrt {2\mathit{var_{\psi_i^l}}} }}$.
\end{appendices}

\end{document}